\documentclass[draftcls, onecolumn, 12pt]{IEEEtran}
\usepackage{amsmath,amssymb,threeparttable,placeins,enumerate,caption}
\usepackage{mathtools,relsize}
\usepackage[pdftex]{graphicx}
\usepackage{epstopdf}
\usepackage{color}
\usepackage[usenames]{xcolor}
\usepackage{amsfonts}
\usepackage{latexsym}
\usepackage{subfigure,multirow,makecell,stfloats}

\usepackage{algorithm}
\usepackage{cite}
\usepackage[noend]{algpseudocode}
\usepackage{setspace}
\usepackage{tabularx}

\def\qi#1 {\fbox {\footnote {\ }}\ \footnotetext { From Qi: {\color{red}#1}}}
\usepackage{amssymb}

\usepackage{amsthm}

\theoremstyle{remark} 
\newtheorem{theorem}{{{\textit{Theorem}}}}

\newtheorem{lemma}{{{\textit{Lemma}}}}
\newtheorem{corollary}{{{{\textit{Corollary}}}}}

\newtheorem{definition}{{{\textit{Definition}}}}

\newtheorem{remark}{{{\textit{Remark}}}}

\newtheorem{example}{{{\textit{Example}}}}

\newtheorem{construction}{{{\textit{Construction}}}}
\title{Doppler Resilient Complementary Sequences: Tighter Aperiodic Ambiguity Function Bound and Optimal Constructions }
	\author{Zheng Wang, Yang Yang,~\IEEEmembership{Member,~IEEE}, Zhengchun Zhou,~\IEEEmembership{Member,~IEEE}, Avik Ranjan Adhikary,~\IEEEmembership{Member,~IEEE}, and Pingzhi Fan,~\IEEEmembership{Fellow,~IEEE}.
	\thanks{ 
		Z. Wang, Y. Yang, and A. R. Adhikary are with the School of Mathematics, Southwest Jiaotong University, Chengdu, 611756, China. (e-mail: wang\_z@my.swjtu.edu.cn, yang\_data@qq.com, avik.adhikary@ieee.org).
		
		Z. Zhou is with the School of Information Science and Technology, Southwest Jiaotong University, Chengdu, 611756, China. (e-mail: zzc@swjtu.edu.cn).

		P. Fan is with the Institute of Mobile Communications, Southwest
		Jiaotong University, Chengdu 611756, China (e-mail: p.fan@ieee.org).
	}
}
 
\begin{document}
\maketitle

\begin{abstract}
Doppler-resilient complementary sequence sets (DRCSs) are crucial in modern communication and sensing systems in mobile environments. In this paper, we propose a new lower bound for the aperiodic ambiguity function (AF) of unimodular DRCSs based on weight vectors, which generalizes the existing bound as a special case. The proposed lower bound is tighter than the Shen-Yang-Zhou-Liu-Fan bound.
Finally, we propose a novel class of aperiodic DRCSs with small alphabets based on quasi-Florentine rectangles and Butson-type Hadamard matrices. Interestingly, the proposed DRCSs asymptotically satisfy the proposed bound. 

\end{abstract}


\section{Introduction}
In the 1960s, Golay introduced Golay complementary pairs (GCPs) \cite{golay1961complementary}, whose aperiodic autocorrelations sum up to zero at each non-zero time-shift. In \cite{tseng1972complementary}, Tseng and Liu extended the idea of GCP to complementary sequence sets (CSs) and introduced mutually orthogonal complementary sequence sets (MOCSS), where the sum of aperiodic autocorrelation at a non-zero time-shift is zero, and the sum of aperiodic cross-correlation at any time-shift is also zero. Due to their ideal correlation properties, CSs/MOCSSs are widely used in various applications, including channel estimation \cite{channeles2001complementary}, spreading codes for asynchronous multicarrier code-division multiple-access (MC-CDMA) communications \cite{mccdma2001multicarrier}, reducing peak-to-average power ratio (PAPR) in orthogonal frequency-division multiplexing (OFDM) systems \cite{PAPR2021encoding}, Doppler-resilient waveform design \cite{Doppler2008doppler}, \cite{Doppler2024complementary}, integrated sensing and communication (ISAC) \cite{ISAC2017ieee}, \cite{ISACl2020doppler}, and more.
However, the size of MOCSS is upper-bounded by the number of constituent sequences, which limits their applicability in systems that require a large number of users. To address this limitation, one approach is to allow the sum of aperiodic correlations of complementary sequences to be a small value, rather than zero. This concept is known as a quasi-complementary sequence set (QCSS) \cite{liu2013tighter}. In \cite{welch1974lower}, Welch proposed a bound for QCSSs using the inner product theorem, and in \cite{liu2013tighter}, Liu, Guan, and Mow presented a tighter aperiodic correlation bound based on the Levenshtein bound \cite{levenshtein1999new}.

However, modern communication systems must address the Doppler effect caused by high mobility. Traditional sequence designs that focus solely on the correlation function are inefficient for current application scenarios. For modern applications such as joint communication and radar systems, sequences must not only mitigate the interference caused by the channel delay but also be resilient to Doppler interference; such sequences are referred to as Doppler-resilient sequences (DRSs) \cite{ye22}. This results to the performance metric being shifted from correlation function to ambiguity function (AF). Several researchers have explored the lower bounds on the maximum AF magnitude. In 2013, Ding \textit{et al.} \cite{ding13} derived the lower bound for the maximum sidelobe of the AF based on the Welch bound. It is important to note that in many practical applications (e.g., \cite{duggal2020doppler,kumari}), the Doppler frequency range and delay range can be significantly smaller than the bandwidth of the transmitted signal and the sequence duration. Motivated by this, Ye \textit{et al.} \cite{ye22} introduced the concept of low/zero ambiguity zones (ZAZ/LAZ) and derived the lower bound for ZAZ/LAZ sequence sets. More recently, Meng \textit{et al.} \cite{meng2024new} proposed tighter theoretical bounds for the aperiodic AF under specific parameters. Additionally, several DRS sets have been designed that satisfy the proposed AF lower bound.

Based on CSs, Pezeshki \textit{et al.} \cite{Doppler2008doppler} designed Doppler-resilient Golay complementary waveforms (DRGCWs) by utilizing GCPs and the Prouhet–Thue–Morse (PTM) sequence for pulse radar applications.
The core idea of DRGCW is to repeatedly transmit GCPs over a pulse sequence by leveraging the structure of the PTM sequence. Consequently, the Taylor expansion of its AF near zero Doppler exhibits higher-order zeros, leading to a Doppler-resilient AF. Subsequently, extensive research has been reported in this direction, such as \cite{tang2014construction,zhu2017range,zhu2019alternative,duggal2020doppler,wang2021quasi}.
Notably, DRGCW assumes that the Doppler shift is significantly smaller than the pulse repetition frequency (PRF) or that the pulse duration is much shorter than the pulse repetition interval (PRI). Based on this assumption, the intra-pulse Doppler shift can be neglected, and only the inter-pulse Doppler shift needs to be considered\cite{duggal2020doppler}. However, in certain cases, this assumption does not hold due to large Doppler shifts. In \cite{shen2024}, Shen \textit{et al.} proposed a novel transmission approach, assuming that the PRI is an integer multiple of the PRF. The authors transmitted a different sequence within each pulse and accumulated AFs across multiple sequences, thereby achieving lower AF sidelobes. Such sequences are referred to as Doppler-resilient complementary sequence sets (DRCSs) in \cite{shen2024}. In \cite{shen2024}, the authors derived lower bounds for the periodic AF, aperiodic AF, and odd periodic AF of DRCSs. Further, the authors in \cite{shen2024} proposed various constructions, including optimal constructions for periodic DRCSs.
 However, the problem of designing sequence sets which achieves the aperiodic AF bound for DRCSs, still remains open.

 
Motivated by the works of \cite{liu2013tighter} and \cite{shen2024}, we propose a new lower bound for the aperiodic auto-AF of DRCSs, which generalizes several previously established bounds, like Levenshtein bound \cite{levenshtein1999new}, Liu-Guan-Mow bound \cite{liu2013tighter}, Meng-Guan-Ge-Liu-Fan bound \cite{meng2024new}, Liu-Zhou-Udaya bound \cite{liu2021tighter}, and Peng-Fan bound \cite{peng2004generalised}, as special cases. Additionally, we analyze three distinct types of weight vectors and their corresponding theoretical aperiodic AF bounds for DRCSs. By selecting an appropriate weight vector, we derive a tighter lower bound than the Shen-Yang-Zhou-Liu-Fan bound \cite{shen2024}. Furthermore, using the quasi-Florentine rectangles proposed by Adhikary \textit{et al.} in \cite{Avik2024} and Butson-type Hadamard matrices, we introduce a class of DRCSs. Interestingly, the proposed DRCSs satisfy the proposed bound asymptotically, thus solving the open problem posed in \cite{shen2024}. Furthermore, the proposed construction yields DRCSs with flexible parameters over any fixed alphabet.

The rest of the paper is organized as follows. Section II introduces the necessary notations and lemmas. In Section III, we present a tighter aperiodic lower bound. Section IV analyzes three types of weight vectors. In Section V, we propose a new class of DRCSs with asymptotic achievability. Finally, Section VI concludes the paper.

\section{Preliminaries}
For convenience, we will use the following notations consistently in this paper:
\begin{itemize}
	\item \( \mathbb{Z}_N \) denotes the ring of integers modulo \( N \).
	\item \( \omega_N = e^{\frac{2 \pi \sqrt{-1}}{N}} \) is a primitive \( N \)-th complex root of unity.
	\item \( (\cdot)^* \), \( (\cdot)^T \), and \( (\cdot)^H \) denote the complex conjugate, transpose, and conjugate transpose, respectively.
	\item \( \mathbb{F}_p \) denotes a finite field containing \( p \) elements, and \( \mathbb{F}_{p^n} \) represents its degree-\( n \) extension.
	\item \( \odot \) denotes the Hadamard product.
	\item \( \text{circ}(\mathbf{a},k) \) represents the circular left shift of \( \mathbf{a} \) by \( k \) elements.
	\item \( \mathbf{0}_{1\times Z} \) denotes the all-zero row vector of size \( 1 \times Z \).
\end{itemize}

\subsection{Ambiguity Function}
Let $\mathbf{a}=(a(0),a(1),\cdots,a(N-1))$ and $\mathbf{b}=(b(0),b(1),$ $\cdots,b(N-1))$ be two complex unimodular sequences with period $N,$ i.e., $|a(i)|=1$ and $0\leq i<N$.
The aperiodic cross AF of $\mathbf{a}$ and $\mathbf{b}$ at time shift $\tau$ and Doppler $v$ at is defined as follows:
\begin{align*}
	AF_{\mathbf{a}, \mathbf{b}}(\tau, v)=\left\{\begin{aligned}
		&\sum_{t=0}^{N-1-\tau} a(t) b^*(t+\tau) \omega_N^{vt}, & 0 \leq \tau \leq N-1,\\
		&\sum_{t=-\tau}^{N-1} a(t) b^*(t+\tau) \omega_N^{vt}, & 1-N \leq \tau<0,\\
		&0,&|\tau|\geq N.
	\end{aligned}\right.
\end{align*}
If $\mathbf{a}=\mathbf{b},$ then $AF_{\mathbf{a}, \mathbf{b}}(\tau, \nu)$ is called aperiodic auto AF and written as $AF_{\mathbf{a}}(\tau,v)$.

A DRCS set $\mathcal{C}=\left\{\mathbf{C}^{(0)}, \mathbf{C}^{(1)}, \cdots, \mathbf{C}^{(K-1)}\right\}$ contains $K$ DRCSs, each of which consists of $M\geq 2$ elementary sequences of length $N$, i.e., $\mathbf{C}^{(k)} =\{\mathbf{c}_0^{(k)}, \mathbf{c}_1^{(k)}, \cdots, \mathbf{c}_{M-1}^{(k)}\}$, where $\mathbf{c}_m^{(k)} =(c_{m}^{(k)}(0), c_{m}^{(k)}(1), \cdots, c_{m}^{(k)}(N-1))$, $0 \leq k < K$ and $0\leq m<M$. 
For two DRCSs $\mathbf{C}^{(k_1)}$ and $\mathbf{C}^{(k_2)}$, their aperiodic cross AF is defined as the aperiodic ambiguity function sum, i.e.,
\begin{align*}
AF_{\mathbf{C}^{(k_1)}, \mathbf{C}^{(k_2)}}(\tau,v)=\sum_{m=0}^{M-1}AF_{\mathbf{c}_m^{(k_1)}, \mathbf{c}_m^{(k_2)}}(\tau, v).
\end{align*}
When $k_1 = k_2$, it is abbreviated to $AF_{\mathbf{C}^{(k_1)}}(\tau, v)$.

\subsection{Low Ambiguity Zone and Existing Aperiodic DRCS Set Bound}
For a DRCS set $\mathcal{C}$, its maximum aperiodic AF magnitude over a region $\Pi =(-Z_x, Z_x) \times(-Z_y, Z_y)\subseteq(-N, N) \times(-N, N)$ is defined as $\theta_{\max}(\mathcal{C})=\max \left\{\theta_a(\mathcal{C}), \theta_c(\mathcal{C})\right\}$, where
\begin{align*}
\theta_a(\mathcal{C})=&\max \left\{\left|A F_{\mathbf{C}} (\tau, v)\right|: \mathbf{C}  \in \mathcal{C},(\tau, v) \neq(0,0) \in \Pi\right\},\\
\theta_c(\mathcal{C})=&\max \left\{\left|A F_{\mathbf{C}, \mathbf{D}}(\tau, f)\right|: \mathbf{C}, \mathbf{D} \in \mathcal{C},(\tau, v) \in \Pi\right\}
\end{align*}
are called the maximum aperiodic auto AF magnitude and the maximum aperiodic cross AF magnitude, respectively.
Then, $\mathcal{C}$ is referred to as an aperiodic $(K, M, N, \theta_{\max},\Pi)$-DRCS set.

The only existing aperiodic AF lower bound of DRCSs in the literature can be expressed as:
\begin{lemma}[\cite{shen2024}]\label{lem-shenbound}
	For a $\left(K, M, N, \theta_{\max }, \Pi\right)$-DRCS set, where $\Pi=\left(-Z_x, Z_x\right) \times\left(-Z_y, Z_y\right), 1 \leq Z_x, Z_y \leq N$, the lower bound of the aperiodic AF is given by
	\begin{align*}
		\theta_{\max } \geq \frac{M N}{\sqrt{Z_y}} \sqrt{\frac{\frac{K Z_x Z_y}{M\left(N+Z_x-1\right)}-1}{K Z_x-1}}.
	\end{align*}
\end{lemma}

In the next section, we further tighten this bound for a class of aperiodic DRCSs. To evaluate the closeness between the theoretical and achieved lower bounds of the aperiodic AF, the optimality factor $\rho$ is defined as follows.

\begin{definition}[Optimality Factor]
	For a $\left(K, M, N, \theta_{\max}, \Pi\right)$-DRCS set, where $\Pi=\left(-Z_x, Z_x\right) \times\left(-Z_y, Z_y\right)$, the aperiodic AF lower bound of DRCSs is denoted as $\theta_{\text{opti}}$, i.e., $\theta_{\max}\geq \theta_{\text{opti}}$. The optimality factor $\rho$ is defined as follows:
	\begin{align*}
		\rho = \frac{\theta_{\max}}{\theta_{\text{opti}}}.
	\end{align*}
\end{definition}
In general, $\rho \geq 1$. We call a DRCS set optimal if $\rho = 1$. A DRCS set is asymptotically optimal if $\rho \rightarrow 1$ as $N \rightarrow +\infty$.

\subsection{Quasi-Florentine Rectangle}
\begin{definition}[\cite{Avik2024}]
A matrix \( \mathbf{A} \) over \( \mathbb{Z}_N \) is called a \textit{quasi-Florentine rectangle} if it satisfies the following two conditions:

\begin{itemize}
	\item \textbf{C1:} Each row of the matrix contains exactly \( N-1 \) distinct symbols, with every symbol appearing exactly once in each row.
	\item \textbf{C2:} For any ordered pair \( (a, b) \) of two distinct symbols and for any integer \( m \) from 1 to \( N-2 \), there is at most one row in which the symbol \( b \) appears exactly \( m \) steps to the right of the symbol \( a \).
\end{itemize}

\end{definition}

The quasi-Florentine rectangle is highly useful for constructing DRCSs, and its construction can be found in \cite{Avik2024}. In the following, we will present a construction of the quasi-Florentine rectangle using finite fields.

\begin{construction}[\cite{Avik2024}]\label{AC1}
Let \( p \) be a prime number, \( n \) be a positive integer, \( f(x) \) be a primitive polynomial of degree \( n \) over \( \mathbb{F}_p \), and let \( \alpha \) be a primitive element of \( \mathbb{F}_{p^n} \). Then the set \( \{1,\alpha, \alpha^2, \dots, \alpha^{n-1}\} \) forms a basis of \( \mathbb{F}_p \) over \( \mathbb{F}_{p^n} \). For any \( \beta \in \mathbb{F}_{p^n} \), there exists a vector \( (a_0, a_1, \dots, a_{n-1}) \in \mathbb{F}_p^n \) such that
\[
\beta = a_0 + a_1 \alpha + \dots + a_{n-1} \alpha^{n-1}.
\]
Define a one-to-one mapping \( \psi \) from \( \mathbb{F}_{p^n} \) to \( \mathbb{Z}_{p^n} \), where \( \psi(0) = 0 \) and
\[
\psi(\beta) = a_0 + a_1 p + \dots + a_{n-1} p^{n-1}.
\]
\end{construction}

\begin{theorem}[\cite{Avik2024}]\label{T1}
	Let $\mathbf{A}$ be a matrix of order \( p^n \times (p^n - 1) \) defined as follows:
\[
\mathbf{A} = \left[\begin{array}{cccc}
	a_{0,0} & a_{0,1} & \cdots & a_{0, p^n-2} \\
	a_{1,0} & a_{1,1} & \cdots & a_{1, p^n-2} \\
	\vdots & \vdots & \ddots & \vdots \\
	a_{p^n-1,0} & a_{p^n-1,1} & \cdots & a_{p^n-1, p^n-2}
\end{array}\right]_{p^n \times (p^n-1)}
\]
 where
\[
a_{i,j} =
\begin{cases}
	\psi(\alpha^j), & i = 0, \\
	\psi(\alpha^j + \alpha^{i-1}), & 0 < i < p^n,
\end{cases}
\]
and $\psi$ is a one-to-one mapping as defined in Construction \ref{AC1}.
Then the matrix $\mathbf{A}$ is a quasi-Florentine rectangle of size \( p^n \times (p^n - 1) \) over $\mathbb{Z}_{p^n}$.
\end{theorem}

\begin{corollary}[\cite{Avik2024}]\label{cor1}
Let $\mathbf{A}$ be the matrix defined in Theorem \ref{T1}. By adding an extra leftmost column where all elements are $p^n$, the modified matrix $\mathbf{A}$ of order $p^n \times p^n$ becomes a quasi-Florentine rectangle over $\mathbb{Z}_{p^n+1}$.
\end{corollary}

 Below we give the important lemma.
\begin{lemma}[\cite{Avik2024}]\label{QFRlemma}
	Let $\mathbf{A}$ be a $K \times(N-1)$ quasi-Florentine rectangle over $\mathbb{Z}_N$, denoted as follows:
	\begin{align*}
		\mathbf{A}=\left[\begin{array}{cccc}
			a_{0,0} & a_{0,1} & \cdots & a_{0, p^n-2} \\
			a_{1,0} & a_{1,1} & \cdots & a_{1, p^n-2} \\
			\vdots & \vdots & \ddots & \vdots \\
			a_{K-1,0} & a_{K-1,1} & \cdots & a_{K-1, p^n-2}
		\end{array}\right],
	\end{align*}
	where $a_{i, j}$ denotes the $j$-th element in the $i$-th row. Let $\pi_k: \mathbb{Z}_N \rightarrow \mathbb{Z}_N$ be a permutation, where $\pi_k$ is equivalent to the $k$-th row of $\mathbf{A}$. For $0 \leq k_1 \neq k_2 \leq K-1$, we have $\pi_{k_1,j}=\pi_{k_2,j+\tau}$ has at most one solution for $0 \leq \tau<N-1$, where $0 \leq j+\tau<N-1$.
\end{lemma}
More details on quasi-Florentine rectangles can be found in \cite{Avik2024}.

\subsection{Butson-type Hadamard Matrix}
\begin{definition}[\cite{had1973complex}]
  Let $N$ and $r$ be two positive integers, and $\mathbf{B}=[\omega_r^{b_{i,j}}]_{i,j=0}^{N-1}$ be a matrix of order $N$, where $b_{i,j} \in \mathbb{Z} (0\leq i,j<N-1)$. If $\mathbf{B}\mathbf{B}^H=N\mathbf{I}$, then $\mathbf{B}$ is called Butson-type Hadamard matrix, denoted by $BH(N,r)$, where $\mathbf{I}$ is the identity matrix of order $N$.
\end{definition}

\begin{remark}
The discrete Fourier transform (DFT) matrices, the Walsh-Hadamard matrices, and the Hadamard matrices are apecial case of Butson-type Hadamard matrices with parameters $BH(N,N)$, $BH(2^m,2)$, and $BH(4N,2),$ respectively. 
\end{remark}

This paper utilizes Butson-type Hadamard matrices to construct DRCSs with smaller alphabets. In \cite{hadc2006complex}, the authors presented Butson-type Hadamard matrices for various parameters, and Table \ref{BH} lists some known parameters of seed Butson-type Hadamard matrices with \( N \leq 21 \) and \( r \leq 10 \). Additionally, Wallis \cite{had1973complex} proposed other Butson-type Hadamard matrices using the Kronecker product.

\begin{table}[htp]
\begin{center}
	\caption{Parameters of seed Butson-type Hadamarad matrices over alphabet size $\leq 10$.}   
	\label{BH}
\begin{tabular}{|c|c|}
	\hline \begin{tabular}{c} 
		Alphabet \\
		size
	\end{tabular} & Parameters \\
	\hline $\mathbb{Z}_2$ & $B H(2,2)$ \\
	\hline $\mathbb{Z}_3$ & $B H(3,3), B H(6,3), B H(12,3),B H(21,3)$ \\
	\hline $\mathbb{Z}_4$ & $B H(4,4), B H(6,4), B H(10,4),B H(12,4), B H(14,4)$ \\
	\hline $\mathbb{Z}_5$ & $B H(5,5), B H(10,5)$ \\
	\hline $\mathbb{Z}_6$ & $B H(6,6), B H(7,6), B H(9,6),BH(10,6),BH(13,6),BH(14,6)$ \\
	\hline $\mathbb{Z}_7$ &$B H(7,7), B H(14,7)$ \\
	\hline $\mathbb{Z}_8$ & $B H(8,8), B H(16,8)$ \\
	\hline $\mathbb{Z}_9$ & $BH(9,9)$ \\
	\hline $\mathbb{Z}_{10}$ & $B H(9,10), B H(10,10), B H(14,10)$ \\
	\hline
\end{tabular}
\end{center}
\end{table}


\section{Proposed Lower Bounds of DRCSs}
In this section, we derive the aperiodic AF lower bound of unimodular DRCSs $\mathcal{C}$. Before deriving the main theorems, let us first give a lemma that show a property of AF in \cite{ye22}.

\begin{lemma}[\cite{ye22}]\label{lemma}
For any unimodular sequence $\mathbf{a}$, its aperiodic auto-AF satisfies
\begin{align}\label{lemma1}
	AF_{\mathbf{a}}(0, v)=0, \quad 0< |v|<N.
\end{align}
\end{lemma}
Next, we derive the aperiodic AF lower bound of  unimodular DRCSs $\mathcal{C}$ in region $\Pi=(-N,N)\times (-Z_y,Z_y)$, where $1\leq Z_y\leq N$. First, based on Frobenius norm of an auto- and cross-ambiguity matrix, we define satisfy the following conditions two weight vectors $\mathbf{w}=\left[w_0, w_1, \cdots ,w_{2N-2}\right]^T$ and $\mathbf{p}=\left[p_0, p_1, \cdots, p_{Z_y-1}\right]^T$, where
\begin{align}
	& \sum_{r=0}^{2N-2}w_r=1,~ w_r\geq 0,  \label{weighttime}\\
	& \sum_{i=0}^{Z_y-1} p_i=1, ~ p_i \geq 0.\label{weightdoppler}
\end{align} 

Let $\mathcal{C}$ be a $(K, M, N, \theta_{\max},\Pi)$-DRCS set, where $\Pi=(-Z_x,Z_y)\times(-Z_y,Z_y)$, $1\leq Z_y<N$, and $Z_x=N$.  Consider the matrix 
\begin{align}\label{matrix}
	\mathbf{U}_{\left(Z_x, Z_y\right)}=\left[\mathbf{U}_0^{\left(Z_x, Z_y\right)}, \mathbf{U}_1^{\left(Z_x, Z_y\right)}, \ldots, \mathbf{U}_{K-1}^{\left(Z_x, Z_y\right)}\right]^T
\end{align}
of size $K(2N-1)Z_y \times M(2N-1)$, where
\begin{align*}
&\mathbf{U}_k^{\left(Z_x, Z_y\right)}=\left[\mathbf{U}_k^{\left(Z_x\right)}(0), \mathbf{U}_k^{\left(Z_x\right)}(1), \ldots, \mathbf{U}_k^{\left(Z_x\right)}\left(Z_y-1\right)\right]^T, 0\leq k\leq K-1,\\
&\mathbf{U}_k^{\left(Z_x\right)}(v_i)=\left[\mathbf{u}_{k,1}^{\left(Z_x\right)}(v_i), \mathbf{u}_{k,2}^{\left(Z_x\right)}(v_i), \ldots, \mathbf{u}_{k,M}^{\left(Z_x\right)}\left(v_i\right)\right]^T, 0\leq v_i< Z_y,
\end{align*}
and 
\begin{align*}
	\mathbf{u}_{k,m}^{\left(Z_x\right)}(v_i)=\left[\begin{array}{c}
		\sqrt{p_i} \sqrt{w_0} \text{circ}\left(\mathbf{u}_{k,m}\odot \mathbf{F}_{v_i},0\right) \\
		\sqrt{p_i} \sqrt{w_1} \text{circ}\left(\mathbf{u}_{k,m}\odot \mathbf{F}_{v_i},1\right) \\
		\vdots \\
		\sqrt{p_i} \sqrt{w_{2N-2}} \text{circ}\left(\mathbf{u}_{k,m}\odot \mathbf{F}_{v_i},2N-2\right)
	\end{array}\right],
\end{align*}
where $0\leq m\leq M-1$, $\mathbf{u}_{k,m}=[\mathbf{c}_m^k,\mathbf{0}_{1\times N-1}]$, $\mathbf{F}_{v_i}=[1,\omega_N^{v_i},\omega_N^{2v_i},\cdots,\omega_N^{(N-1)v_i},\mathbf{0}_{1\times N-1}]$. Based on the properties of matrix operations, the following identity holds:
\begin{align} \label{uhu}
	\left\|\mathbf{U}_{\left(Z_x, Z_y\right)}^H \mathbf{U}_{\left(Z_x, Z_y\right)}\right\|_F^2=\left\|\mathbf{U}_{\left(Z_x, Z_y\right)} \mathbf{U}_{\left(Z_x, Z_y\right)}^H\right\|_F^2,
\end{align}
where $\|\cdot\|_F$ is the Frobenius norm.

\begin{lemma}\label{Lobound}
Considering the right side of  (\ref{uhu}), the following inequality holds:
	\begin{align*}
		\left\|\mathbf{U}_{\left(Z_x, Z_y\right)} \mathbf{U}_{\left(Z_x, Z_y\right)}^H\right\|_F^2\geq MK^2 (N-\sum_{s, t=0}^{2N-2}\tau_{s,t,N}w_s w_t), 
	\end{align*}
	where $\tau_{s,t,N}=\min\{|t-s|,2N-1-|t-s|\}$.
\end{lemma}
\begin{proof}
By definition, we have
\begin{align*}
\begin{split}
\left\|\mathbf{U}_{\left(Z_x, Z_y\right)} \mathbf{U}_{\left(Z_x, Z_y\right)}^H\right\|_F^2 
&=\sum_{m=0}^{M-1}\sum_{s,t=0}^{2N-2}\left|\sum_{k=0}^{K-1}\sum_{i=0}^{Z_y-1}\sum_{r=0}^{2N-2} w_rp_{i}u_{k,m}(s+r)u_{k,m}^{*}(t+r)e^{\frac{j 2 \pi i\left(s-t\right)}{N}}\right|^2\\
&\geq M \sum_{s=0}^{2N-2}\left|\sum_{j=0}^{K-1}\sum_{i=0}^{Z_y-1}\sum_{r=0}^{2N-2}  w_rp_{i}u_{k,0}(s+r)u_{k,0}^{*}(s+r)\right|^2\\
&=MK^2\sum_{s=0}^{2N-2}\left|\sum_{r=0}^{2N-2}w_ru_{0,0}(s+r)u_{0,0}^{*}(s+r)\right|^2  \\
&=MK^2 (N-\sum_{s, t=0}^{2N-2}\tau_{s,t,N}w_s w_t), 
\end{split}
\end{align*}
where $s+r$ is calculated over $\mathbb{Z}_{2N-1}$, $u_{k,m}(s)=0$ for $N\leq s\leq 2N-2$.
\end{proof}

\begin{lemma}\label{Upbound}
Considering the left side of  (\ref{uhu}), the following inequality holds:
	\begin{align}\label{Upbound_1}
		\left\|\mathbf{U}_{\left(Z_x, Z_y\right)}^H \mathbf{U}_{\left(Z_x, Z_y\right)}\right\|_F^2\leq \frac{KM^2N^2}{Z_y}\sum_{r=0}^{2N-2}w_r^2+K\theta_a^2-K\theta_a^2\sum_{r=0}^{2N-2}w_r^2+K(K-1)\theta_c^2 
	\end{align}	
	and 
	\begin{align}\label{Upbound_2}
		\left\|\mathbf{U}_{\left(Z_x, Z_y\right)}^H \mathbf{U}_{\left(Z_x, Z_y\right)}\right\|_F^2\leq K^2\theta_{\max}^2+\frac{KM^2N^2}{Z_y}\sum_{r=0}^{2N-2}w_r^2-K\theta_{\max}^2\sum_{r=0}^{2N-2}w_r^2.
	\end{align}
\end{lemma}
\begin{proof}
	The left-hand side term of (\ref{uhu}) can be expanded as follows:
\begin{align*}
\begin{split}
\left\|\mathbf{U}_{\left(Z_x, Z_y\right)}^H \mathbf{U}_{\left(Z_x, Z_y\right)}\right\|_F^2
&=\sum_{k,k'=0}^{K-1}\sum_{r,r'=0}^{2N-2}\sum_{i,i'=0}^{Z_y-1}|
AF_{\mathbf{u}_k,\mathbf{u}_{k'}}(\tau_{r,r'},v_{i,i'})|^2p_ip_{i'}w_rw_{r'}\\ &=\sum_{k=0}^{K-1}\sum_{r,r'=0}^{2N-2}\sum_{i,i'=0}^{Z_y-1}|
AF_{\mathbf{u}_k}(\tau_{r,r'},v_{i,i'})|^2p_ip_{i'}w_rw_{r'}\\ 
&+\sum_{k,k'=0,k\neq k_1}^{K-1}\sum_{r,r'=0}^{2N-2}\sum_{i,i'=0}^{Z_y-1}
\left|AF_{\mathbf{u}_k,\mathbf{u}_{k'}}(\tau_{r,r'},v_{i,i'})\right|^2p_ip_{i'}w_rw_{r'}.
\end{split}
\end{align*}
By Lemma \ref{lemma}, we have
\begin{align*}
\begin{split}
&\sum_{k=0}^{K-1}\sum_{r,r'=0}^{2N-2}\sum_{i,i'=0}^{Z_y-1}
|AF_{\mathbf{u}_{k}}(\tau_{r,r'},v_{i,i'})|^2p_ip_{i'}w_rw_{r'}\\
=&K\sum_{r,r'=0}^{2N-2}\sum_{i,i'=0}^{Z_y-1}|AF_{\mathbf{u}_{k}}(\tau_{r,r'},v_{i,i'})|^2p_ip_{i_1}w_rw_{r_1}\\
&+KM^2N^2\sum_{r=0}^{2N-2}\sum_{i=0}^{Z_y-1}p_i^2w_r^2 -K\sum_{i,i'=0}^{Z_y-1}\sum_{r=0}^{2N-2}|AF_{\mathbf{u}_{k}}(\tau_{r},v_{i,i'})|^2 p_ip_{i'}w_r^2.
\end{split}
\end{align*}
Therefore, we have
\begin{align}\label{righth}
\begin{split}
\left\|\mathbf{U}_{\left(Z_x, Z_y\right)}^H \mathbf{U}_{\left(Z_x, Z_y\right)}\right\|_F^2
&\leq  KM^2N^2\sum_{r=0}^{2N-2}\sum_{i=0}^{Z_y-1}p_i^2w_r^2+K\sum_{r,r'=0}^{2N-2}\sum_{i,i'=0}^{Z_y-1}
\theta_a^2 p_ip_{i'}w_rw_{r'}\\ 
&-K\sum_{i,i'=0}^{Z_y-1}\sum_{r=0}^{2N-2}\theta_a^2 p_ip_{i'}w_r^2
+K(K-1)\sum_{r,r'=0}^{2N-2}\sum_{i,i'=0}^{Z_y-1}\theta_c^2p_ip_{i'}w_rw_{r'}\\ &=KM^2N^2\sum_{r=0}^{2N-2}\sum_{i=0}^{Z_y-1}p_i^2w_r^2
+K\theta_a^2-K\sum_{r=0}^{2N-2}\theta_a^2w_r^2+K(K-1)\theta_c^2.
\end{split}
\end{align}
According to  (\ref{weightdoppler}) and Cauchy-Schwarz inequality, we can determine that the minimum value
$\sum_{r=0}^{Z_y-1} p_r^2$ is $\frac{1}{Z_y}$, and this minimum value is achieved if and only if 
\begin{align}\label{wdoppler}
 p_r=\frac{1}{Z_y},0\leq r\leq Z_y-1.
\end{align}
Hence, (\ref{wdoppler}) is the optimal choice of weight vector $\mathbf{p}$ for (\ref{righth}), as it yields the tightest upper bound of 
\begin{align*}
\left\|\mathbf{U}_{\left(Z_x, Z_y\right)}^H \mathbf{U}_{\left(Z_x, Z_y\right)}\right\|_F^2 
& \leq \frac{KM^2N^2}{Z_y}\sum_{r=0}^{2N-2}w_r^2+K\theta_a^2-K\sum_{r=0}^{2N-2}\theta_a^2w_r^2
+K(K-1)\theta_c^2 \nonumber \\
&\leq K^2\theta_{\max}^2+\frac{KM^2N^2}{Z_y}\sum_{r=0}^{2N-2}w_r^2
-K\theta_{\max}^2\sum_{r=0}^{2N-2}w_r^2.
\end{align*}
This completes the proof.
\end{proof}

Next, we present the main conclusion of this paper based on Lemma \ref{Lobound} and Lemma \ref{Upbound}.
\begin{theorem}\label{T2}
For a unimodular $(K, M, N, \theta_{\max}, \Pi)$-DRCS set, where $\Pi = (-N, N) \times (-Z_y, Z_y)$ and $1 \leq Z_y \leq N$, the lower bounds of the aperiodic AF are given by 
\begin{align}\label{T1a}
	\left(1-\sum_{r=0}^{2N-2}w_r^2\right)\theta_a^2+(K-1)\theta_c^2 \geq KM\left(N-\sum_{s, t=0}^{2N-2}\tau_{s,t,N}w_s w_t\right) -\frac{M^2N^2}{Z_y}\sum_{r=0}^{2N-2}w_r^2
\end{align}
and
\begin{align}\label{eq2}
	\theta_{\max}^2\geq M\left(N-\frac{\mathbf{Q}_{2N-1}\left(\mathbf{w},\frac{N(MN-Z_y)}{KZ_y}\right)}{\left(1-\frac{1}{K}\sum_{r=0}^{2N-2}w_r^2\right)}\right),
\end{align}
where the quadratic form
$\mathbf{Q}_{2N-1}(\mathbf{w}, a) \triangleq \mathbf{w}	\mathbf{Q}_{2N-1} \mathbf{w}^T
=a \sum_{i=0}^{2N-2} w_i^2+\sum_{s, t=0}^{2N-2} \tau_{s, t, N} w_s w_t$, 
$\mathbf{Q}_{2N-1}$ is a $(2N-1) \times (2N-1)$ matrix whose diagonal entries are equal to $a$, and for $s \neq t$, its $(s,t)$-th entry is $\tau_{s, t, N}$.
\end{theorem}
\begin{proof}
From (\ref{uhu}), Lemma \ref{Lobound}, and (\ref{Upbound_1}) in Lemma \ref{Upbound}, we have

\begin{align*}
 \frac{KM^2N^2}{Z_y}\sum_{r=0}^{2N-2}w_r^2+K\theta_a^2-K\theta_a^2\sum_{r=0}^{2N-2}w_r^2+K(K-1)\theta_c^2 \geq	MK^2 (N-\sum_{s, t=0}^{2N-2}\tau_{s,t,N}w_s w_t).
\end{align*}
Further, we have 
\begin{align*}
	\left(1-\sum_{r=0}^{2N-2}w_r^2\right)\theta_a^2+(K-1)\theta_c^2 \geq KM\left(N-\sum_{s, t=0}^{2N-2}\tau_{s,t,N}w_s w_t\right) -\frac{M^2N^2}{Z_y}\sum_{r=0}^{2N-2}w_r^2.
\end{align*}

From (\ref{uhu}), Lemma \ref{Lobound}, and (\ref{Upbound_2}) in Lemma \ref{Upbound}, we have
\begin{align*}
	K^2\theta_{\max}^2+\frac{KM^2N^2}{Z_y}\sum_{r=0}^{2N-2}w_r^2-K\theta_{\max}^2\sum_{r=0}^{2N-2}w_r^2\geq MK^2 (N-\sum_{s, t=0}^{2N-2}\tau_{s,t,N}w_s w_t).
\end{align*}
Further, we have 
\begin{align*}
	(K-\sum_{r=0}^{2N-2}w_r^2)\theta_{\max}^2\geq MK (N-\sum_{s, t=0}^{2N-2}\tau_{s,t,N}w_s w_t)-\frac{M^2N^2}{Z_y}\sum_{r=0}^{2N-2}w_r^2
\end{align*}
So, we have
\begin{align*}
\theta_{\max}^2&\geq M\left(\frac{(N-\sum_{s, t=0}^{2N-2}\tau_{s,t,N}w_s w_t)-\frac{MN^2}{KZ_y}\sum_{r=0}^{2N-2}w_r^2}{\left(1-\frac{1}{K}\sum_{r=0}^{2N-2}w_r^2\right)}\right)\\
&=M\left(\frac{(N-\frac{N}{K}\sum_{r=0}^{2N-2}w_r^2+\frac{N}{K}\sum_{r=0}^{2N-2}w_r^2-\sum_{s, t=0}^{2N-2}\tau_{s,t,N}w_s w_t)-\frac{MN^2}{KZ_y}\sum_{r=0}^{2N-2}w_r^2}{\left(1-\frac{1}{K}\sum_{r=0}^{2N-2}w_r^2\right)}\right)\\
&=M\left(\frac{(N-\frac{N}{K}\sum_{r=0}^{2N-2}w_r^2-\sum_{s, t=0}^{2N-2}\tau_{s,t,N}w_s w_t)-\frac{N(MN-Z_y)}{KZ_y}\sum_{r=0}^{2N-2}w_r^2}{\left(1-\frac{1}{K}\sum_{r=0}^{2N-2}w_r^2\right)}\right)\\
&=M\left(N-\frac{\mathbf{Q}_{2N-1}\left(\mathbf{w},\frac{N(MN-Z_y)}{KZ_y}\right)}{\left(1-\frac{1}{K}\sum_{r=0}^{2N-2}w_r^2\right)}\right).
\end{align*}
This completes the proof.
\end{proof}

\begin{remark}
	It is straightforward to see that Theorem \ref{T2} encompasses several classic theoretical bounds as special cases. For instance, when \( Z_y = 1 \), the proposed lower bound reduces to the Liu-Guan-Mow bound for QCSSs in \cite{liu2013tighter}. Furthermore, when \( Z_y = 1 \) and \( M = 1 \), it simplifies to the Levenshtein bound for aperiodic correlation in \cite{levenshtein1999new}.
\end{remark} 

In the following, we derive the lower bound of aperiodic AF of unimodular DRCSs $\mathcal{C}$ in region $\Pi=(-Z_x,Z_x)\times (-Z_y,Z_y)$, where $1\leq Z_x,Z_y\leq N$. In this case, choose a special weight vector, i.e., $\mathbf{w}=(w_0,w_1,\cdots,w_{Z_x-1},\mathbf{0}_{1\times(2N-Z_x-1)}).$ Therefore, the size of matrix $\mathbf{U}_{(Z_x, Z_y)}$ in  (\ref{matrix}) becomes $KZ_xZ_y \times M(2N-1),$ the weighting condition (\ref{weighttime}) can be replaced by
\begin{align*}
	\sum_{r=0}^{Z_x-1} w_r=1,~ w_r \geq 0.
\end{align*}

Similar to the derivation of Theorem \ref{T2} based on Frobenius norm,  we can get the following theorem.
\begin{theorem}\label{T3}
	For a unimodular $(N,K, M, \theta_{\max},\Pi)$-DRCS set, where $\Pi=(-Z_x,Z_x)\times(-Z_y,Z_y)$ and $1\leq Z_x,Z_y<N,$ the lower bounds of the aperiodic AF are given by 
\begin{align}
		\left(1-\sum_{r=0}^{Z_x-1}w_r^2\right)\theta_a^2+(K-1)\theta_c^2 \geq KM\left(N-\sum_{s, t=0}^{Z_x-1}\tau_{s,t,N}w_s w_t\right) -\frac{M^2N^2}{Z_y}\sum_{r=0}^{Z_x-1}w_r^2,
\end{align}
\begin{align}\label{T12}
\theta_{\max}^2\geq M\left(N-\frac{\mathbf{Q}_{Z_x}\left(\mathbf{w},\frac{N(MN-Z_y)}{KZ_y}\right)}
{\left(1-\frac{1}{K}\sum_{r=0}^{Z_x-1}w_r^2\right)}\right) \nonumber \\
\geq M\left(N-\mathbf{Q}_{Z_x}\left(\mathbf{w},\frac{MN^2}{KZ_y}\right)\right).
\end{align}	
\end{theorem}

\begin{proof}
The proof of Theorem \ref{T3} follows a similar approach to that of Theorem \ref{T2} and is therefore omitted for brevity.
\end{proof}


\begin{remark}
It is worth noting that when different parameters are selected, Theorem \ref{T3} can degenerate into some classic theoretical bounds, such as Meng-Guan-Ge-Liu-Fan bound \cite{meng2024new}, Liu-Zhou-Udaya bound \cite{liu2021tighter}, and Peng-Fan bound \cite{peng2004generalised}.
\end{remark}


\section{Discussions on the Proposed Lower Bounds for DRCSs}
In the previous section, we derived the theoretical lower bound on the sidelobes of the aperiodic AF for DRCSs, which is closely related to the choice of weight vectors. In this section, we discuss three specific weight vectors previously used in \cite{levenshtein1999new} and \cite{liu2013tighter}, and demonstrate that the proposed theoretical lower bound is tighter than the Shen-Yang-Zhou-Liu-Fan bound.

Next, we consider a uniform weight vector proposed by \cite{levenshtein1999new}, defined as  
\begin{align}\label{weightvection1}
	w_r=\frac{1}{2N-1}, \quad 0\le r\le 2N-2.
\end{align}
By applying the weight vector in (\ref{weightvection1}) to Theorem \ref{T2}, we obtain an improved lower bound for DRCSs as follows.

\begin{corollary}
For weight vector  (\ref{weightvection1}), the lower bounds of the aperiodic AF are given by
\begin{align}\label{C11}
	\frac{2N}{2N-1}\theta_a^2+(K-1)\theta_c^2 \geq M^2N^2\frac{KZ_y/M-1}{(2N-1)Z_y}
\end{align}
and
\begin{align}\label{C12}
	\theta_{\max}^2\geq M^2N^2\frac{KZ_y/M-1}{[K(2N-1)-1]Z_y}.
\end{align}
\end{corollary}

\begin{proof}
By applying the weight vector (\ref{weightvection1}) to (\ref{T1a}), we can transform (\ref{T1a}) into  
\begin{align*}
	\left(1-\frac{1}{2N-1}\right)\theta_a^2+(K-1)\theta_c^2 \geq KMN\left(N-\frac{N(N-1)}{2N-1}\right) -\frac{M^2N^2}{Z_y}\frac{1}{2N-1}.
\end{align*}
Consequently, (\ref{C11}) and (\ref{C12}) follow as direct results.
\end{proof}

\begin{corollary}\label{cor-kmnzy}
If $K$, $M$, $N$, and $Z_y$ satisfy the following relation:  
\begin{align*}
	K \leq \left\lfloor \frac{4(M N-Z_y) N}{Z_y} \sin^2 \frac{\pi}{2(2N-1)}\right\rfloor,
\end{align*}  
then the theoretical bound in (\ref{C12}) cannot be further improved. In particular, as $N$ approaches positive infinity, we have  
\begin{align*}
	K \leq \left\lfloor \frac{\pi^2 M}{4Z_y} \right\rfloor.
\end{align*}
\end{corollary}

\begin{proof}
Before proving Corollary \ref{cor-kmnzy}, let us first review two key results from Berlekamp \cite{berlekamp2015algebraic}:

\begin{itemize}
	\item[1)] The principal eigenvalue and the secondary eigenvalues of the quadratic matrix $\mathbf{Q}_{2N-1}$ are given by:
	\begin{align*}
		\lambda_0 &= a + (N-1)N, \\
		\lambda_k &= a - \frac{1 - (-1)^k \cos \frac{\pi k}{2N-1}}{2 \sin^2 \frac{\pi k}{2N-1}},
	\end{align*}
	where $1 \leq k \leq 2N-2$.
	\item[2)] The quadratic function $Q_{2N-1}(\mathbf{w}, a)$ is convex if all the secondary eigenvalues $\lambda_k \geq 0$. In this case, the global minimum is achieved when $\mathbf{w} = \frac{1}{2N-1}(1,1, \dots, 1)$. Furthermore, it is evident that the term $1-\frac{1}{K} \sum_{i=0}^{2N-2} w_i^2$ attains its global maximum when $\mathbf{w} = \frac{1}{2N-1}(1,1, \dots, 1)$. Therefore, the right-hand side of (\ref{eq2}) reaches its global maximum under this choice of $\mathbf{w}$.
\end{itemize}

Based on the above two results, the bound (\ref{C12}) cannot be further improved if
\begin{align*}
	\min _{1 \leq k \leq 2N-2} \lambda_k 
	=  \frac{(M N - Z_y) N}{K Z_y} - \frac{1}{4 \sin^2 \frac{\pi}{2(2N-1)}}
	\geq  0,
\end{align*}
which is equivalent to
\begin{align*}
	K \leq \left\lfloor \frac{4(MN - Z_y) N}{Z_y} \sin^2 \frac{\pi}{2(2N-1)}\right\rfloor.
\end{align*}
This completes the proof.
\end{proof}

In the previous, the weight vector was divided into $2N-1$ equal parts, ensuring that all time-shifts were taken into account. In this subsection, we consider another uniform weight vector proposed by Levenshtein \cite{levenshtein1999new}, which is defined as

\begin{align}\label{wv1}
	w_r= 
	\begin{cases}
		\frac{1}{m}, & r \in\{0,1, \ldots, m-1\},  \\ 
		0, & r \in\{m, m+1, \ldots, Z_x-1\},
	\end{cases}
\end{align}
where $1\leq m\leq Z_x\leq N$.

\begin{corollary} 
	For weight vector  (\ref{wv1}), the lower bound of the aperiodic AF is given by
	\begin{align}\label{Wv1}
      \theta_{\max}^2 \geq \frac{3mKMNZ_y-KM(m^2-1)Z_y-3M^2N^2}{3(mK-1)Z_y}.
	\end{align}
\end{corollary}

\begin{proof}
By definition, it is obvious that
\begin{align*}
      \sum_{r=0}^{2N-2}w_r^2=\frac{1}{m}.
\end{align*}
Hence,  (\ref{eq2}) can be written as
\begin{align*}
	\theta_{\max}^2(K-\frac{1}{m})&\geq KMN-KM\sum_{s, t=0}^{m-1}|s-t| w_s w_t-M^2N^2\frac{1}{Z_ym} \nonumber \\
	&=KMN-KM\frac{(m-1)(m+1)}{3m}-M^2N^2\frac{1}{Z_ym}.
\end{align*}	
Further, the above inequality can be converted to
\begin{align*}
	\theta_{\max}^2 \geq  \frac{3mKMNZ_y-KM(m^2-1)Z_y-3M^2N^2}{3(mK-1)Z_y}.
\end{align*}
This completes the proof.
\end{proof}

The value of the right-hand side of the inequality (\ref{Wv1}) is determined by selecting an appropriate $m$, which in turn leads to the following theoretical lower bound.

\begin{corollary} \label{LevAFc1}
For \( K > \frac{3M}{Z_y} \) and \( N \sqrt{\frac{3M}{KZ_y}} \leq Z_x \leq N \), the lower bounds of the aperiodic AF are given by
\begin{align}
	(1 - \frac{KZ_y}{3MN^2}) \theta_a^2 + (K - 1) \theta_c^2 \geq KM \left(N - 2N \sqrt{\frac{M}{3KZ_y}}\right),
\end{align}
and
\begin{align}\label{C1eq}
	\theta_{\max}^2 \geq MN \left( 1 - 2 \sqrt{\frac{M}{3KZ_y}} \right).
\end{align}
\end{corollary}
\begin{proof}
By the weight vector in  (\ref{wv1}), we have (\ref{T1a}) can be reduced as
\begin{align*}
	(1-\frac{1}{m})\theta_a^2+(K-1)\theta_c^2 \geq KMN-KM\frac{(m-1)(m+1)}{3m}-\frac{M^2N^2}{Z_ym}.
\end{align*}
If $m$ arounds $N\sqrt{\frac{3M}{KZ_y}}$, and let $\hat{m}=N\sqrt{\frac{3M}{KZ_y}}-\epsilon$ and $0\leq \epsilon<1$, then
\begin{align*}
\begin{split}
	(1-\frac{1}{m})\theta_a^2+(K-1)\theta_c^2 &\geq 	(1-\frac{1}{\hat{m}})\theta_a^2+(K-1)\theta_c^2\\
	&\geq KMN-KM\frac{\hat{m}^2-1}{3\hat{m}}-\frac{M^2N^2}{Z_y\hat{m}}\\
	&\geq KM(N-2N\sqrt{\frac{M}{3KZ_y}}).
\end{split}
\end{align*}
So, we have 
\begin{align*}
	(1-\frac{KZ_y}{3MN^2})\theta_a^2+(K-1)\theta_c^2 \geq 	
	KM(N-2N\sqrt{\frac{M}{3KZ_y}}).
\end{align*}
Further, the above equation can be written as
\begin{align*}
	K\theta_{\max}^2 \geq (1-\frac{KZ_y}{3MN^2})\theta_a^2+(K-1)\theta_c^2 \geq 	
	KM(N-2N\sqrt{\frac{M}{3KZ_y}}),
\end{align*}
which is equivalent to
\begin{align*}
	\theta_{\max}^2  \geq 	
	MN(1-2\sqrt{\frac{M}{3KZ_y}}).
\end{align*}
This completes the proof.
\end{proof}

From (\ref{T12}), it can be seen that we aim for the inequality on the right side to be as small as possible, which is equivalent to minimizing \( \mathbf{Q}_{Z_x} (\mathbf{w}, \frac{MN^2}{KZ_y}) \). Here, we apply the method of Lagrange multipliers to find the minimum value of \( \mathbf{Q}_{Z_x} (\mathbf{w}, \frac{MN^2}{KZ_y}) \).Additionally, a weight vector can be obtained by minimizing the following function using the Lagrange multiplier method proposed by Levenshtein \cite{levenshtein1999new}:
\begin{align*}
	f(\mathbf{w}, \lambda) = \mathbf{Q}_{Z_x}\left( \mathbf{w}, \frac{MN^2}{KZ_y} \right) - 2 \lambda \left( \sum_{r=0}^{Z_x-1} w_r^2 - 1 \right),
\end{align*}
where \( 2 \leq m \leq Z_x \leq N \). Let \( \gamma = \arccos \left( 1 - \frac{KZ_y}{MN^2} \right) \), \( KZ_y \leq MN^2 \), and \( m \) be an even number such that \( m \gamma \leq \pi + \gamma \). Setting \( \gamma_0 = \frac{\pi - m \gamma + \gamma}{2} \), by relating the quadratic minimization solution of \( f(\mathbf{w}, \lambda) \) to the Chebyshev polynomials of the second kind \cite{levenshtein1999new}, we can obtain the following weight vector:
\begin{align}\label{Wv2}
	w_r = 
	\begin{cases}
		\frac{\sin \frac{\gamma}{2}}{\sin \frac{m \gamma}{2}} \sin \left( \gamma_0 + r \gamma \right), & 0 \leq r \leq m-1, \\
		0, & m \leq r \leq 2N-2.
	\end{cases}
\end{align}

The weight vector  (\ref{Wv2}) is applied to Theorem \ref{T3}, a new improved lower bound of DRCSs is obtained as follows.
\begin{corollary}\label{cor-eqsinw}
For a unimodular $(K,M,N,\theta_{\max},\Pi)$-DRCS set with $K Z_y \leq MN^2$, and any integer $m$ satisfying $1 \leq m<$ $\min \{Z_x+1, \pi/\gamma+1\}$, the lower bound of the aperiodic AF is given by
\begin{align*}
	\theta_{\max }^2 & \geq \frac{\theta_c^2(K-1)+\theta_a^2}{K} \nonumber \\
	& \geq NM-\frac{m-1}{2}-\frac{\sin \frac{m \gamma}{2}-\sin \frac{(m-2) \gamma}{2}}{2(1-\cos \gamma) \sin \frac{m \gamma}{2}}.
\end{align*}
\end{corollary}

\begin{proof}
	The proof is similar to Lemma 2 in \cite{liu2013new}, which is ignored here.
\end{proof}

Obviously, when $Z_x>\pi/\gamma$, one can minimize $f(\mathbf{w},\lambda)$ over different values of $m$ by setting $m=\left\lfloor\pi/\gamma\right\rfloor+1$. Thus, Corollary \ref{cor-eqsinw} follows immediately.

\begin{corollary}\label{C5}
For a LAZ satisfying $Z_x>\pi/\gamma$ and $5M \leq K Z_y \leq MN^2$, the lower bound of the aperiodic AF is given by
\begin{align*}
\theta_{\max }^2 & \geq \frac{\theta_c^2(K-1)+\theta_a^2}{K} \nonumber\\
	& \geq M\left(N-\left\lceil\frac{\pi N}{\sqrt{8 K Z_y/M}}\right\rceil\right).
\end{align*}
\end{corollary}
\begin{proof}
	Let $m=\left\lfloor\pi/\gamma\right\rfloor+1$ for Corollary \ref{cor-eqsinw}, then the proof is similar to Corollary 4 in \cite{levenshtein1999new}, which is ignored here.
\end{proof}

\section{Asymptotically Optimal Construction of DRCSs} \label{5}
In this section, we propose a new aperiodic DRCSs which are asymptotically optimal with respect to proposed bound in  (\ref{C1eq}). 

\begin{construction}\label{c1}
For any positive integer $ N \geq 2$, let $\mathbf{A}=[a_{i,j}]$ be a quasi-Florentine rectangle of size $K\times (N-1)$ over $\mathbb{Z}_N$. Additionally, let $\pi_p$ denote the $p$-th row of $\mathbf{A}$ and $\pi_p(n)$ be $n$-th element of $\pi_p$. Finally, let $\mathbf{B}$ be a Butson-type Hadamard matrix of order $N$ over the alphabet $\mathbb{Z}_r$, given by
\begin{align*}
	 \mathbf{B}=\left[\begin{array}{c}
		\mathbf{b}_0 \\
		\mathbf{b}_1 \\
		\vdots \\
		\mathbf{b}_{N-1}
	\end{array}\right]=\left[\begin{array}{cccc}
		\omega_r^{b_{0,0}} & \omega_r^{b_{0,1}} & \cdots & \omega_r^{b_{0, N-1}} \\
		\omega_r^{b_{1,0}} & \omega_r^{b_{1,1}} & \cdots & \omega_r^{b_{1, N-1}} \\
		\vdots & \vdots & \ddots & \vdots \\
		\omega_r^{b_{N-1,0}} & \omega_r^{b_{N-1,1}} & \cdots & \omega_r^{b_{N-1, N-1}}
	\end{array}\right].
\end{align*}
Define a DRCS set $\mathcal{C}=\{\mathbf{C}^{(0)},\mathbf{C}^{(1)},\cdots,\mathbf{C}^{(K-1)}\}$, where $\mathbf{C}^{(k)} =\{\mathbf{c}_0^{(k)}, \mathbf{c}_1^{(k)}, \cdots, \mathbf{c}_{N-1}^{(k)}\}$, $\mathbf{c}_m^{(k)} = (c_{m, 0}^{(k)}, c_{m, 1}^{(k)}, \cdots, c_{m, N-2}^{(k)})$, and
\begin{align*}
  c_{m, n}^{(k)}=\omega_r^{b_{\pi_k(n),m}}
\end{align*}
for $0\leq k<K$, $0 \leq m<N$, $0 \leq n<N-1$.
\end{construction}

\begin{theorem}\label{thc1}
The sequence set $\mathcal{C}$ constructed using Construction \ref{c1} is an aperiodic DRCS set with parameters  $\left(K, N, N-1, N, \Pi\right)$-DRCS, where $\Pi=(-N+1,N-1)\times (-N+1, N-1)$. 
\end{theorem}

\begin{proof}
According to the definition of DRCSs, we divide the proof of Theorem \ref{thc1} into two cases: auto-AF and cross-AF.

Case 1 (auto-AF): For any $0 \leq k<K$ and $(\tau,v)\ne(0,0)$, we have
\begin{align*}
\begin{split}
AF_{\mathbf{C}^{(k)}}(\tau, v) & =\sum_{m=0}^{N-1} \sum_{n=0}^{N-2-\tau} 
c_{m, n}^{(k)} (c_{m, n+\tau}^{(k)})^{*}\omega_{N-1}^{nv}\\ 
&=\sum_{m=0}^{N-1}\sum_{n=0}^{N-2-\tau}
\omega_r^{b_{\pi_k(n),m}}\omega_r^{-b_{\pi_k(n+\tau),m}}\omega_{N-1}^{nv}\\
&=\sum_{n=0}^{N-2-\tau}\omega_{N-1}^{nv}\sum_{m=0}^{N-1}\omega_r^{b_{\pi_k(n),m}-b_{\pi_k(n+\tau),m}}.
\end{split}
\end{align*}
Based on Lemma \ref{QFRlemma}, $|AF_{\mathbf{C}^{(k)}}(\tau, v)|=0$ is obvious.

Case 2 (cross-AF): For any $0 \leq k_1\neq k_2<K$, we have
\begin{align*}
\begin{split}
A F_{\mathbf{C}^{(k_1)},\mathbf{C}^{(k_2)}}(\tau, v) & 
=\sum_{m=0}^{N-1} \sum_{n=0}^{N-2-\tau} 
c_{m, n}^{(k_1)} (c_{m, n+\tau}^{(k_2)})^{*}\omega_{N-1}^{nv}\\ 
&=\sum_{m=0}^{N-1}\sum_{n=0}^{N-2-\tau}
\omega_r^{b_{\pi_{k_1}(n),m}}\omega_r^{-b_{\pi_{k_2}(n+\tau),m}}\omega_{N-1}^{nv}\\
&=\sum_{n=0}^{N-2-\tau}\omega_{N-1}^{nv}\sum_{m=0}^{N-1}\omega_r^{b_{\pi_{k_1}(n),m}-b_{\pi_{k_2}(n+\tau),m}}.
\end{split}
\end{align*}
According to Lemma \ref{QFRlemma}, we have $\pi_{k_1}(n)=\pi_{k_2}(n+\tau)$ for $0 \leq n \leq n+\tau<N-1$ and $0 \leq k_1 \neq k_2<K$, which has at most one solution. If it does not have a solution, then
\begin{align*}
	A F_{\mathbf{C}^{(k_1)},\mathbf{C}^{(k_2)}}(\tau, v)=0.
\end{align*}
If it has one solution, denoted as $n'$. The vectors $\mathbf{b}_{\pi_{k_1}(n)}$ and $\mathbf{b}_{\pi_{k_2}(n+\tau)}$ are two distinct rows of a Butson-type Hadamard matrix, and they are orthogonal to each other. Therefore, we have
\begin{align*}
	A F_{\mathbf{C}^{(k_1)},\mathbf{C}^{(k_2)}}(\tau, v)
	&=N\omega_{N-1}^{n'v}+\sum_{n=0,n\neq n'}^{N-2-\tau}\omega_{N-1}^{nv}
	\sum_{m=0}^{N-1}\omega_r^{b_{\pi_{k_1}(n),m}-b_{\pi_{k_2}(n+\tau),m}} \nonumber \\
	&=N\omega_{N-1}^{n'v}.
\end{align*}
Namely, $|AF_{\mathbf{C}^{(k_1)},\mathbf{C}^{(k_2)}}(\tau, v)|=N$ for all $k_1 \neq k_2$ and $0 \leq \tau, v<N-1$. This completes the proof.
\end{proof}

\begin{remark}
Note that Theorem \ref{thc1} depends on the availability of quasi-Florentine rectangles. Based on Theorem \ref{T1} and Corollary \ref{cor1}, there exist quasi-Florentine rectangles of size $p^n\times (p^n-1)$ and $p^n\times p^n$, where $p$ is a prime number and $n\ge1$. More details on quasi-Florentine rectangles can be found in \cite{Avik2024}.
\end{remark}

In the following, we show an example to illustrate the proposed construction.

\begin{example}\label{ex1}
	Let $p=3$, $n=2$, $N=p^2=9$ and $f(x)=x^2+2x+2$ be a primitive polynomial, also let $\alpha$ be a primitive element of $\mathbb{F}_{9}$ with $\alpha^2+2\alpha+2=0$, then $\mathbb{F}_{9}=\{0,\alpha^0,\alpha^1,\alpha^2,\alpha^3,\cdots,\alpha^{7}\}$. Based on the Corollary \ref{cor1}, the quasi-Florentine rectangle of size $9 \times 9$ over $\mathbb{Z}_{10}$ denoted by
	\begin{equation*}
		\mathbf{A}={\left[\begin{array}{ccccccccc} 
				1 & 3 & 4 & 7 & 2 & 6 & 8 & 5 & 9 \\
				2 & 4 & 5 & 8 & 0 & 7 & 6 & 3 & 9 \\
				4 & 6 & 7 & 1 & 5 & 0 & 2 & 8 & 9 \\
				5 & 7 & 8 & 2 & 3 & 1 & 0 & 6 & 9 \\
				8 & 1 & 2 & 5 & 6 & 4 & 3 & 0 & 9 \\ 
				0 & 5 & 3 & 6 & 1 & 8 & 7 & 4 & 9 \\
				7 & 0 & 1 & 4 & 8 & 3 & 5 & 2 & 9 \\
				6 & 2 & 0 & 3 & 7 & 5 & 4 & 1 & 9 \\
				3 & 8 & 6 & 0 & 4 & 2 & 1 & 7 & 9 \\
			\end{array}\right].}
	\end{equation*}
Let $\mathbf{B}$ be $BH(10,5)$ matrix of order $10$ over the alphabet $\mathbb{Z}_5$, as following:
\begin{equation*}
	\mathbf{B}=\left[\begin{array}{rrrrrrrrrr}
0 & 0 & 0 & 0 & 0 & 0 & 0 & 0 & 0 & 0 \\
0 & 1 & 2 & 3 & 4 & 4 & 0 & 1 & 2 & 3 \\
0 & 2 & 4 & 1 & 3 & 1 & 3 & 0 & 2 & 4 \\
0 & 3 & 1 & 4 & 2 & 1 & 4 & 2 & 0 & 3 \\
0 & 4 & 3 & 2 & 1 & 4 & 3 & 2 & 1 & 0 \\
0 & 3 & 2 & 2 & 3 & 0 & 1 & 4 & 4 & 1 \\
0 & 2 & 0 & 4 & 4 & 3 & 2 & 3 & 1 & 1 \\
0 & 1 & 3 & 1 & 0 & 2 & 4 & 3 & 4 & 2 \\
0 & 0 & 1 & 3 & 1 & 2 & 2 & 4 & 3 & 4 \\
0 & 4 & 4 & 0 & 2 & 3 & 1 & 1 & 3 & 2
	\end{array}\right].
\end{equation*}
Note that each element in $\mathbf{B}$ represents a power of $\omega_5$.

 A glimpse of the aperiodic auto-AF and cross-AF of the sequences in $\mathcal{C}$ can be seen in Fig. \ref{fig:ex1}. According to Theorem \ref{thc1}, we can obtain an aperiodic $(9,10,9,10,\Pi)$-DRCS set over an alphabet of size 5, where $\Pi=(-9,9)\times (-9,9)$. The DRCSs $\mathbf{C}^{(0)}$, $\mathbf{C}^{(1)}$,$\mathbf{C}^{(2)}$, and $\mathbf{C}^{(3)}$ in Example \ref{ex1}, are
 given in Table \ref{DRCS_sequence}. Note that each element represents a power
 of $\omega_5$.
  
  \begin{figure}[htp]
  	\centering
  	\includegraphics[width=1\linewidth]{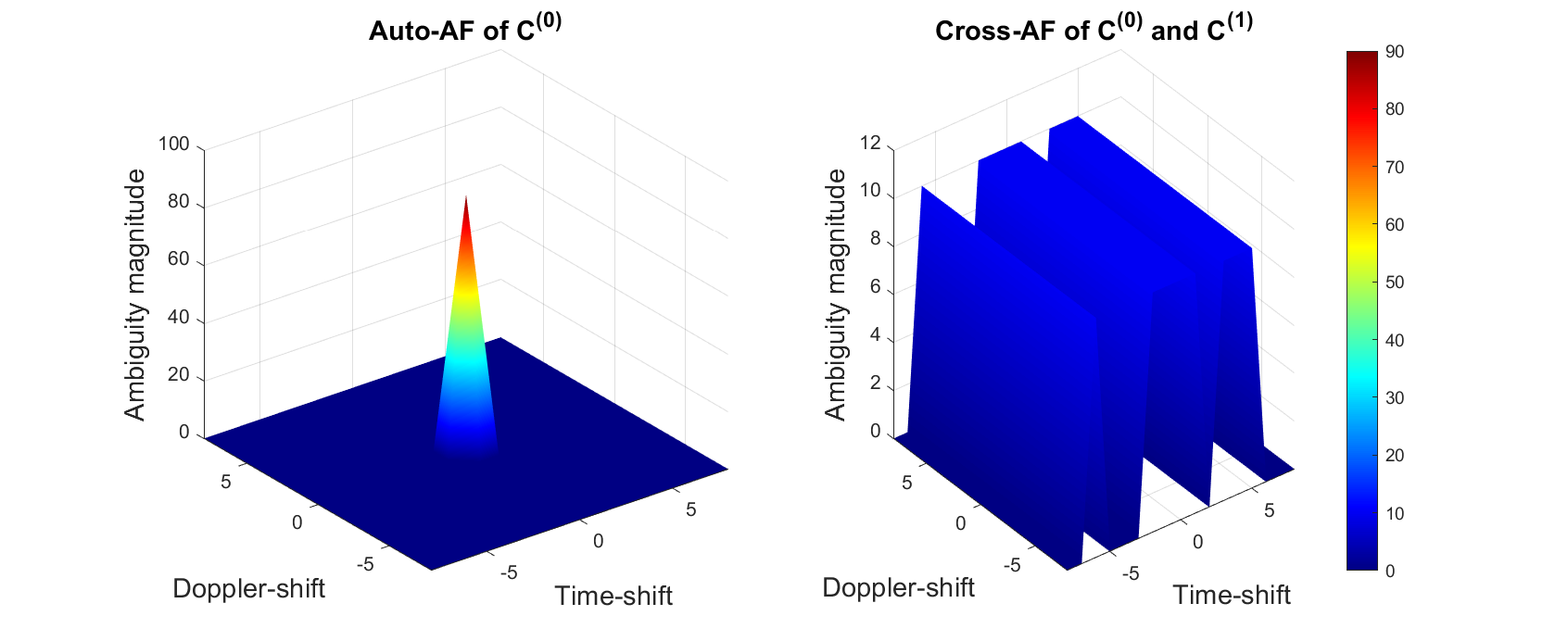}
  	\caption{A glimpse of the aperiodic auto-AF and cross-AF of the sequence set $\mathcal{C}$ in Example \ref{ex1}.}
  	\label{fig:ex1}
  \end{figure}
  
 \begin{table}[htp]
 	\begin{center}
 		\caption{DRCSs in Example \ref{ex1}.}   
 			\label{DRCS_sequence}
 		\begin{tabular}{|c|c|}
 			\hline $\mathbf{C}^{(0)}$ & $\mathbf{C}^{(1)}$\\
 			\hline \begin{tabular}{cccccccccc}
0 & 1 & 2 & 3 & 4 & 4 & 0 & 1 & 2 & 3 \\
0 & 3 & 1 & 4 & 2 & 1 & 4 & 2 & 0 & 3 \\
0 & 4 & 3 & 2 & 1 & 4 & 3 & 2 & 1 & 0 \\
0 & 1 & 3 & 1 & 0 & 2 & 4 & 3 & 4 & 2 \\
0 & 2 & 4 & 1 & 3 & 1 & 3 & 0 & 2 & 4 \\
0 & 2 & 0 & 4 & 4 & 3 & 2 & 3 & 1 & 1 \\
0 & 0 & 1 & 3 & 1 & 2 & 2 & 4 & 3 & 4 \\
0 & 3 & 2 & 2 & 3 & 0 & 1 & 4 & 4 & 1 \\
0 & 4 & 4 & 0 & 2 & 3 & 1 & 1 & 3 & 2
 			\end{tabular}
 			& \begin{tabular}{cccccccccc}
 		0 & 2 & 4 & 1 & 3 & 1 & 3 & 0 & 2 & 4 \\
 		0 & 4 & 3 & 2 & 1 & 4 & 3 & 2 & 1 & 0 \\
 		0 & 3 & 2 & 2 & 3 & 0 & 1 & 4 & 4 & 1 \\
 		0 & 0 & 1 & 3 & 1 & 0 & 2 & 4 & 3 & 4 \\
 		0 & 0 & 0 & 0 & 0 & 0 & 0 & 0 & 0 & 0 \\
 		0 & 1 & 3 & 0 & 2 & 2 & 4 & 3 & 0 & 2 \\
 		0 & 2 & 0 & 4 & 3 & 3 & 2 & 3 & 1 & 1 \\
 		0 & 3 & 1 & 2 & 1 & 4 & 4 & 2 & 0 & 3 \\
 		0 & 4 & 4 & 0 & 2 & 3 & 1 & 1 & 3 & 2
 			\end{tabular} \\ 	\hline
 		  $\mathbf{C}^{(2)}$ & $\mathbf{C}^{(3)}$\\ 	\hline
 		 \begin{tabular}{cccccccccc} 
0 & 4 & 3 & 2 & 1 & 4 & 3 & 2 & 1 & 0 \\
0 & 2 & 0 & 4 & 4 & 3 & 2 & 3 & 1 & 1 \\
0 & 1 & 3 & 1 & 0 & 2 & 4 & 3 & 4 & 2 \\
0 & 1 & 2 & 3 & 1 & 0 & 4 & 4 & 2 & 3 \\
0 & 3 & 2 & 2 & 3 & 0 & 1 & 4 & 4 & 1 \\
0 & 0 & 0 & 0 & 0 & 0 & 0 & 0 & 0 & 0 \\
0 & 2 & 4 & 0 & 1 & 0 & 3 & 2 & 1 & 3 \\
0 & 0 & 1 & 3 & 1 & 2 & 2 & 4 & 3 & 4 \\
0 & 4 & 4 & 0 & 2 & 3 & 1 & 1 & 3 & 2
 			\end{tabular} &   \begin{tabular}{cccccccccc} 
0 & 3 & 2 & 2 & 3 & 0 & 1 & 4 & 4 & 1 \\
0 & 1 & 3 & 1 & 0 & 2 & 4 & 3 & 4 & 2 \\
0 & 0 & 1 & 3 & 1 & 2 & 2 & 4 & 3 & 4 \\
0 & 2 & 4 & 1 & 3 & 1 & 3 & 0 & 2 & 4 \\
0 & 3 & 1 & 4 & 2 & 1 & 4 & 2 & 0 & 3 \\
0 & 1 & 2 & 3 & 4 & 4 & 0 & 1 & 2 & 3 \\
0 & 0 & 0 & 0 & 0 & 0 & 0 & 0 & 0 & 0 \\
0 & 2 & 0 & 4 & 4 & 3 & 2 & 3 & 1 & 1 \\
0 & 4 & 4 & 0 & 2 & 3 & 1 & 1 & 3 & 2
 		\end{tabular} \\
 			\hline
 		\end{tabular}
 	\end{center}
 \end{table}
\end{example}

\begin{corollary}\label{cFR}
	The construction proposed in Theorem \ref{thc1} also works when replacing the quasi-Florentine rectangles with Florentine rectangles. By employing a Butson-type Hadamard matrix \( BH(N, r) \) and a Florentine rectangle of order \( K \times N \), we can construct a \( (K, N, N, N, \Pi) \)-DRCS, where \( \Pi = (-N, N) \times (-N, N) \). 
\end{corollary}
\begin{proof}
	The proof is similar to the proof of Theorem \ref{thc1}, hence omitted.
\end{proof}
For details on the construction of Florentine rectangles, please refer to \cite{avik2021asymptotically} and the references therein.

\begin{remark}\label{remark5}
	Notably, when the Butson-type Hadamard matrix is selected as the DFT matrix, the proposed \( (K, N, N, N, \Pi) \)-DRCS reduces to the construction by Shen \textit{et al.} in \cite{shen2024}, where \( \Pi = (-N, N) \times (-N, N) \). Furthermore, in the construction by Shen \textit{et al.} \cite{shen2024}, the alphabet size is \( N \), whereas in Construction \ref{c1}, the alphabet size is \( r \geq 2\). Thus, our proposed DRCSs feature a smaller alphabet size.
\end{remark}

\begin{example}\label{ex2}
Consider \( BH(16,2) \) and a Florentine rectangle with \( K = 16 \). By applying Construction \ref{c1} and Corollary \ref{cFR}, we obtain a \( (16, 16, 16, 16, \Pi) \)-DRCS over \( \mathbb{Z}_2 \), where \( \Pi = (-16,16) \times (-16,16) \). A visualization of the AF is presented in Fig. \ref{fig:ex3}, while the corresponding sequence sets \( \mathbf{C}^{(0)} \) and \( \mathbf{C}^{(1)} \) are listed in Table \ref{DRCS_sequence1}. Note that each element represents a power
of $\omega_2$.

\begin{figure}[htp]
	\centering
	\includegraphics[width=1\linewidth]{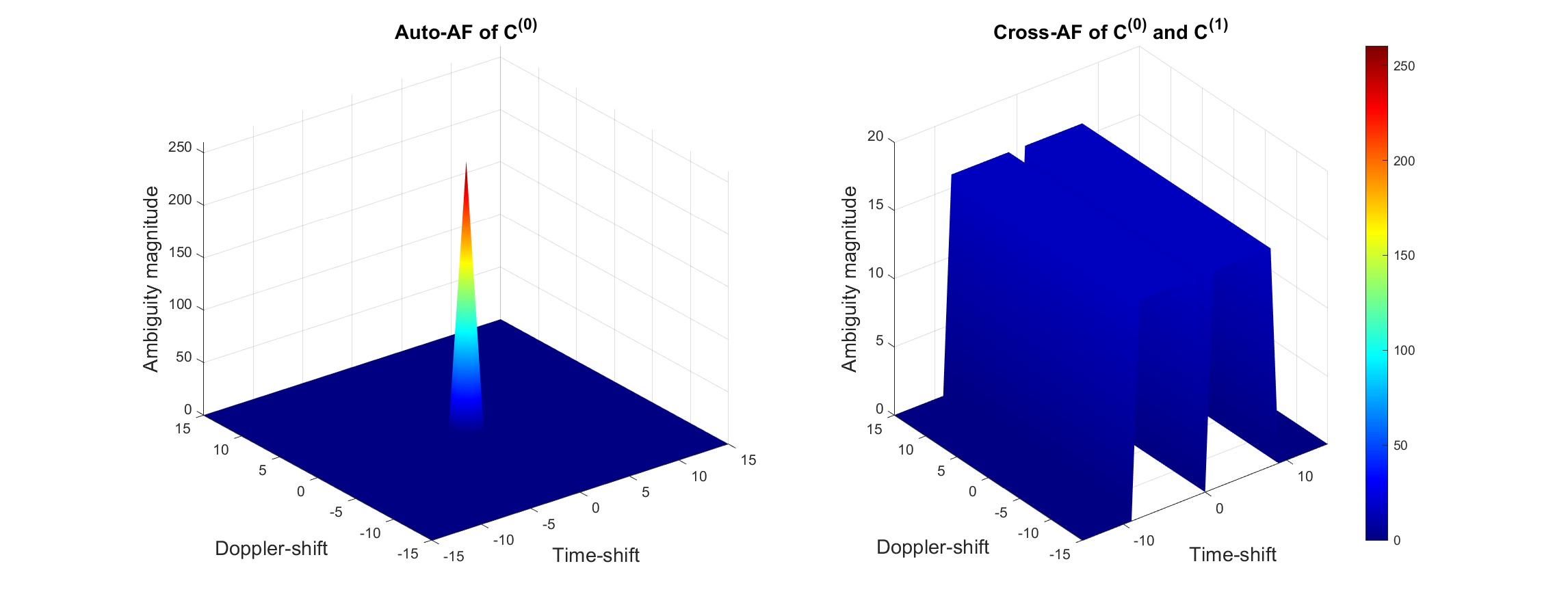}
	\caption{A glimpse of the aperiodic auto-AF and cross-AF of the sequence set $\mathcal{C}$ in Example \ref{ex2}.}
	\label{fig:ex3}
\end{figure}
 \begin{table}[htp]
 	
	\begin{center}
		\caption{DRCSs in Example \ref{ex2}.}   
		\label{DRCS_sequence1}
		\resizebox{\textwidth}{!}{
		\begin{tabular}{|c|c|}
			\hline  $\mathbf{C}^{(0)}$ &	$\mathbf{C}^{(1)}$ \\
			\hline    	  \begin{tabular}{cccccccccccccccc}
				0 & 0 & 0 & 0 & 0 & 0 & 0 & 0 & 0 & 0 & 0 & 0 & 0 & 0 & 0 & 0 \\
				0 & 1 & 0 & 1 & 0 & 1 & 0 & 1 & 0 & 1 & 0 & 1 & 0 & 1 & 0 & 1 \\
				0 & 0 & 1 & 1 & 0 & 0 & 1 & 1 & 0 & 0 & 1 & 1 & 0 & 0 & 1 & 1 \\
				0 & 1 & 1 & 0 & 0 & 1 & 1 & 0 & 0 & 1 & 1 & 0 & 0 & 1 & 1 & 0 \\
				0 & 0 & 0 & 0 & 1 & 1 & 1 & 1 & 0 & 0 & 0 & 0 & 1 & 1 & 1 & 1 \\
				0 & 1 & 0 & 1 & 1 & 0 & 1 & 0 & 0 & 1 & 0 & 1 & 1 & 0 & 1 & 0 \\
				0 & 0 & 1 & 1 & 1 & 1 & 0 & 0 & 0 & 0 & 1 & 1 & 1 & 1 & 0 & 0 \\
				0 & 1 & 1 & 0 & 1 & 0 & 0 & 1 & 0 & 1 & 1 & 0 & 1 & 0 & 0 & 1 \\
				0 & 0 & 0 & 0 & 0 & 0 & 0 & 0 & 1 & 1 & 1 & 1 & 1 & 1 & 1 & 1 \\
				0 & 1 & 0 & 1 & 0 & 1 & 0 & 1 & 1 & 0 & 1 & 0 & 1 & 0 & 1 & 0 \\
				0 & 0 & 1 & 1 & 0 & 0 & 1 & 1 & 1 & 1 & 0 & 0 & 1 & 1 & 0 & 0 \\
				0 & 1 & 1 & 0 & 0 & 1 & 1 & 0 & 1 & 0 & 0 & 1 & 1 & 0 & 0 & 1 \\
				0 & 0 & 0 & 0 & 1 & 1 & 1 & 1 & 1 & 1 & 1 & 1 & 0 & 0 & 0 & 0 \\
				0 & 1 & 0 & 1 & 1 & 0 & 1 & 0 & 1 & 0 & 1 & 0 & 0 & 1 & 0 & 1 \\
				0 & 0 & 1 & 1 & 1 & 1 & 0 & 0 & 1 & 1 & 0 & 0 & 0 & 0 & 1 & 1 \\
				0 & 1 & 1 & 0 & 1 & 0 & 0 & 1 & 1 & 0 & 0 & 1 & 0 & 1 & 1 & 0 
			\end{tabular} & \begin{tabular}{cccccccccccccccc}
				0 & 0 & 0 & 0 & 0 & 0 & 0 & 0 & 0 & 0 & 0 & 0 & 0 & 0 & 0 & 0 \\
				1 & 1 & 1 & 1 & 1 & 1 & 1 & 1 & 0 & 0 & 0 & 0 & 0 & 0 & 0 & 0 \\
				0 & 1 & 0 & 1 & 0 & 1 & 0 & 1 & 0 & 1 & 0 & 1 & 0 & 1 & 0 & 1 \\
				1 & 0 & 1 & 0 & 1 & 0 & 1 & 0 & 0 & 1 & 0 & 1 & 0 & 1 & 0 & 1 \\
				0 & 0 & 1 & 1 & 0 & 0 & 1 & 1 & 0 & 0 & 1 & 1 & 0 & 0 & 1 & 1 \\
				1 & 1 & 0 & 0 & 1 & 1 & 0 & 0 & 0 & 0 & 1 & 1 & 0 & 0 & 1 & 1 \\
				0 & 1 & 1 & 0 & 0 & 1 & 1 & 0 & 0 & 1 & 1 & 0 & 0 & 1 & 1 & 0 \\
				1 & 0 & 0 & 1 & 1 & 0 & 0 & 1 & 0 & 1 & 1 & 0 & 0 & 1 & 1 & 0 \\
				0 & 0 & 0 & 0 & 1 & 1 & 1 & 1 & 0 & 0 & 0 & 0 & 1 & 1 & 1 & 1 \\
				1 & 1 & 1 & 1 & 0 & 0 & 0 & 0 & 0 & 0 & 0 & 0 & 1 & 1 & 1 & 1 \\
				0 & 1 & 0 & 1 & 1 & 0 & 1 & 0 & 0 & 1 & 0 & 1 & 1 & 0 & 1 & 0 \\
				1 & 0 & 1 & 0 & 0 & 1 & 0 & 1 & 0 & 1 & 0 & 1 & 1 & 0 & 1 & 0 \\
				0 & 0 & 1 & 1 & 1 & 1 & 0 & 0 & 0 & 0 & 1 & 1 & 1 & 1 & 0 & 0 \\
				1 & 1 & 0 & 0 & 0 & 0 & 1 & 1 & 0 & 0 & 1 & 1 & 1 & 1 & 0 & 0 \\
				0 & 1 & 1 & 0 & 1 & 0 & 0 & 1 & 0 & 1 & 1 & 0 & 1 & 0 & 0 & 1 \\
				1 & 0 & 0 & 1 & 0 & 1 & 1 & 0 & 0 & 1 & 1 & 0 & 1 & 0 & 0 & 1 
		\end{tabular} \\ 
			\hline
		\end{tabular}}
	\end{center}
\end{table}
\end{example}

In Table \ref{SA} we give few specific examples of the parameters of DRCS over small alphabets constructed using Florentine and quasi-Florentine rectangles. The computation of the optimality factor \( \rho \) is based on the bound given in Corollary \ref{LevAFc1}.

\begin{table}[htp]
	\caption{Few Specific Examples of DRCSs Over Small Alphabets.}  
	\label{SA}
	\resizebox{\textwidth}{!}{
	\begin{tabular}{|c|c|c|c|c|c|c|c|c|c|}
		\hline \begin{tabular}{c} 
			Alphabet \\
			size
		\end{tabular} & $K$ &$M$&$N$&$\theta_{\max}$& $Z_x$& $Z_y$& $\rho$ &\begin{tabular}{c} 
			Butson-type \\
			Hadamard matrix
		\end{tabular}   & \begin{tabular}{c} 
			Florentine rectangle or \\
			quasi-Florentine rectangle
		\end{tabular}  \\
		\hline \multirow{2}{*}{$\mathbb{Z}_2$} & $16$ & $16$ &$16$&$16$&$16$&$16$ &1.1857 & $B H(16,2)$  & Florentine rectangle for $K=16$ \\
		\cline { 2 - 10 } & $127$ &$128$ & $127$& $128$& $127$&$127$& 1.0599 &  $B H(128,2)$   & quasi-Florentine rectangle for $K=128$\\
		\hline \multirow{2}{*}{$\mathbb{Z}_3$} & $71$ &$72$ & $71$& $72$& $71$&$71$ &1.0846 & $B H(72,3)$  & quasi-Florentine rectangle for $K=72$ \\
		\cline { 2 - 10 } & $126$&$126$&$126$&$126$&$126$& $126$& 1.0558  &  $B H(126,3)$  & Florentine rectangle for $K=126$\\
		\hline \multirow{2}{*}{$\mathbb{Z}_4$} &$28$&$28$&$28$&$28$&$28$&$28$ &1.1310 & $B H(28,4)$ & Florentine rectangle for $K=28$ \\
		\cline { 2 - 10 } &$139$&$140$&$139$&$140$&$139$&$139$ &1.0569 & $B H(140,4)$ & quasi-Florentine rectangle for $K=140$ \\
		\hline \multirow{2}{*}{$\mathbb{Z}_5$} &$49$&$50$&$49$&$50$&$49$& $49$& 1.1065 &  $B H(50,5)$ & quasi-Florentine rectangle for $K=50$\\
		\cline { 2 - 10 } &$100$&$100$&$100$&$100$&$100$&$100$& 1.0633 &  $B H(100,5)$ & Florentine rectangle for $K=100$ \\
		\hline \multirow{2}{*}{$\mathbb{Z}_6$}&$89$&$90$&$89$&$90$&$89$ &$89$&1.0739 &  $B H(90,6)$ & quasi-Florentine rectangle for $K=90$ \\
		\cline { 2 - 10 } & $42$&$42$&$42$&$42$&$42$&$42$& 1.1031  & $B H(42,6)$ & Florentine rectangle for $K=42$ \\
		\hline \multirow{2}{*}{$\mathbb{Z}_7$} &$97$&$98$&$97$&$98$&$97$&$97$& 1.0702 &  $B H(98,7)$ & quasi-Florentine rectangle for $K=98$ \\
		\cline { 2 - 10 } &$196$ &$196$&$196$&$196$&$196$&$196$& 1.0440 & $B H(196,7)$ &Florentine rectangle for $K=196$ \\
		\hline \multirow{2}{*}{$\mathbb{Z}_8$} &$127$&$128$&$127$&$128$&$127$&$127$ & 1.0599 &  $B H(128,8)$ & quasi-Florentine rectangle for $K=128$ \\
		\cline { 2 - 10 } &$112$&$112$&$112$&$112$&$112$&$112$& 1.0595 & $B H(112,8)$ & Florentine rectangle for $K=112$ \\
		\hline \multirow{2}{*}{$\mathbb{Z}_9$} &$53$&$54$&$53$&$54$&$53$&$53$&1.1014 &  $B H(54,9)$ & quasi-Florentine rectangle for $K=54$ \\
		\cline { 2 - 10 } &$108$&$108$&$108$&$108$&$108$&$108$&1.0607  &  $B H(108,9)$ & Florentine rectangle for $K=108$ \\
		\hline \multirow{2}{*}{$\mathbb{Z}_{10}$} &$71$&$72$&$71$&$72$&$71$&$71$&1.0846  & $B H(72,10)$ & quasi-Florentine rectangle for $K=72$ \\
		\cline { 2 - 10 } &$70$&$70$&$70$&$70$&$70$&$70$&1.0771 &  $B H(70,10)$ & Florentine rectangle for $K=70$ \\
		\hline
	\end{tabular}}
\end{table}

\section{Comparisons with the Existing Aperiodic AF Lower Bound}
In this section, we evaluate the tightness of the proposed aperiodic AF lower bounds. In addition, we provide some numerical examples for comparison.

For sufficiently large \( K \), and under the conditions \( K > \frac{3M}{Z_y} \) and \( N \sqrt{\frac{3M}{KZ_y}} \leq Z_x \leq N \), Corollary \ref{LevAFc1} shows that \( \theta_{\max}^2 \) approaches \( M N \). Under the condition \( Z_x > \pi/\gamma \) and \( 5M \leq K Z_y \leq M N^2 \), Corollary \ref{C5} also shows that \( \theta_{\max}^2 \) tends to \( M N \). In contrast, the aperiodic AF bound for DRCSs in Lemma \ref{lem-shenbound} indicates that \( \theta_{\max}^2 \) approaches \( \frac{M N}{2} \) when \( Z_x = N \). This means that for larger values of \( K \), \( N \), and \( Z_x = N \), the theoretical lower bound given in Corollary \ref{LevAFc1} improves by nearly a factor of \( \sqrt{2} \) compared to the original lower bound \cite{shen2024} for aperiodic DRCSs.

Next, we analyze how close Construction \ref{c1} is to the theoretical bound derived in Corollary \ref{LevAFc1}.

\begin{theorem}\label{Th5}
	The DRCSs proposed in Theorem \ref{thc1} is asymptotically optimal with respect to the
	bound proposed in  (\ref{C1eq}) when $K\geq 4$ and $N\geq 10$.
\end{theorem}

\begin{proof}
Since \( K \geq 4 \) and \( N \geq 10 \), we have the inequality \( \sqrt{\frac{3MN^2}{KZ_y}} = N\sqrt{\frac{3N}{K(N-1)}} \leq N-1 = Z_x \). According to (\ref{C1eq}) in Corollary \ref{LevAFc1}, the optimality factor is given by
\begin{align}\label{rho1}
\rho = \frac{\theta_{\max}}{\sqrt{MN(1 - 2\sqrt{\frac{M}{3KZ_y}})}} = \frac{N}{\sqrt{N(N-1)\left(1 - 2\sqrt{\frac{N}{3K(N-1)}}\right)}} \rightarrow 1 \quad \text{as} \quad N \rightarrow +\infty.
\end{align}

Thus, the DRCS set \( \mathcal{C} \) asymptotically achieves the theoretical bound presented in Corollary \ref{LevAFc1} when \( K \geq 4 \) and \( N \geq 10 \).
\end{proof}

\begin{remark}
	By the aperiodic lower bound of DRCSs proposed in \cite{shen2024}, as shown in Lemma \ref{lem-shenbound}, the optimality factor
	\begin{align}
		\rho=\frac{\theta_{\max }}{\frac{M N}{\sqrt{Z_y}} \sqrt{\frac{\frac{K Z_x Z_y}{M\left(N+Z_x-1\right)}-1}{K Z_x-1}}}\rightarrow\sqrt{2}, ~N \rightarrow +\infty
	\end{align}
	for DRCS $\mathcal{C}$ in Construction \ref{c1}. From Theorem \ref{Th5}, we have $\rho\rightarrow1,~N \rightarrow +\infty.$ So, this can also show that our proposed theoretical bound is tighter.
\end{remark}
\begin{corollary}
	The DRCSs proposed in Corollary \ref{cFR} is asymptotically optimal with respect to the
     bound proposed in  (\ref{C1eq}) when $K\geq 4$.
\end{corollary}
\begin{proof}
	This proof is similar to Theorem 1. It is omitted here.
\end{proof}
\begin{remark}
The DRCS sequence sets proposed in \cite{shen2024} can be regarded as a special case of Corollary \ref{cFR} based on Remark \ref{remark5}. Therefore, the construction proposed by Shen \textit{et al.} \cite{shen2024} is also asymptotically optimal with respect to the bound given in (\ref{C1eq}) when $K\geq 4$.  
\end{remark}
In the following, we provide numerical examples to demonstrate that the proposed aperiodic AF lower bound is tighter than the existing lower bound presented in Lemma \ref{lem-shenbound}.

In Table \ref{2} we list down some parameters of the proposed aperiodic DRCSs. And we compared the optimality factors (\ref{rho1}) and \cite{shen2024} for $(p^n, p^n, p^n-1, p^n, \Pi)$-DRCS, where $3 \leq p < 120$, $n=2$, and $\Pi=(-p^n+1,p^n-1)\times (-p^n+1,p^n-1)$ in Fig. \ref{fig:ex2}.
 As can be seen from Table \ref{2} and Fig. \ref{fig:ex2}, for the construction of Theorem \ref{thc1}, the Shen-Yang-Zhou-Liu-Fan bound of the aperiodic AF for DRCSs asymptotically reaches $\sqrt{2}$ \cite{shen2024}, and the theoretical bound proposed by Corollary \ref{LevAFc1} asymptotically reaches 1. This can also show that our proposed theoretical bound is tighter.

\begin{table}[htp]
\caption{Some parameters of the proposed aperiodic DRCSs}
\label{2}
\begin{center}
\begin{tabular}{|c|c|c|c|c|c|c|c|}
\hline
     $K$& $M$ & $N$ &$Z_x$ & $Z_y$ & $\theta_{\max}$    & $\rho$ in \cite{shen2024} & $\rho$ in (\ref{rho1}) \\
\hline 9 &9 & 8 & 8& 8& 9  &1.4927 & 1.2521 \\
\hline 25 & 25 & 24 & 24 &24 & 25&1.4322 & 1.1243 \\
\hline49 & 49 & 48 & 48 &48& 49 & 1.4224& 1.0853 \\
\hline  81 & 81 & 80 & 80& 80& 81 & 1.4189& 1.0654 \\
\hline  121 & 121 & 120 & 120& 120& 121  & 1.4173& 1.0532\\
\hline 169 & 169 & 168 & 168 &168& 169  &1.4164& 1.0448 \\
\hline  289 & 289 & 288 & 288& 288& 289 &1.4155& 1.0341 \\
\hline 361 & 361 & 360 & 360 &360& 361  &1.4152& 1.0305 \\
\hline  529 &529 & 528 & 528 &528& 529  &1.4149& 1.0252 \\
\hline
\end{tabular}
\end{center}
\end{table}	
 \begin{figure}[htp]
 	\centering
 	\includegraphics[width=0.7\linewidth]{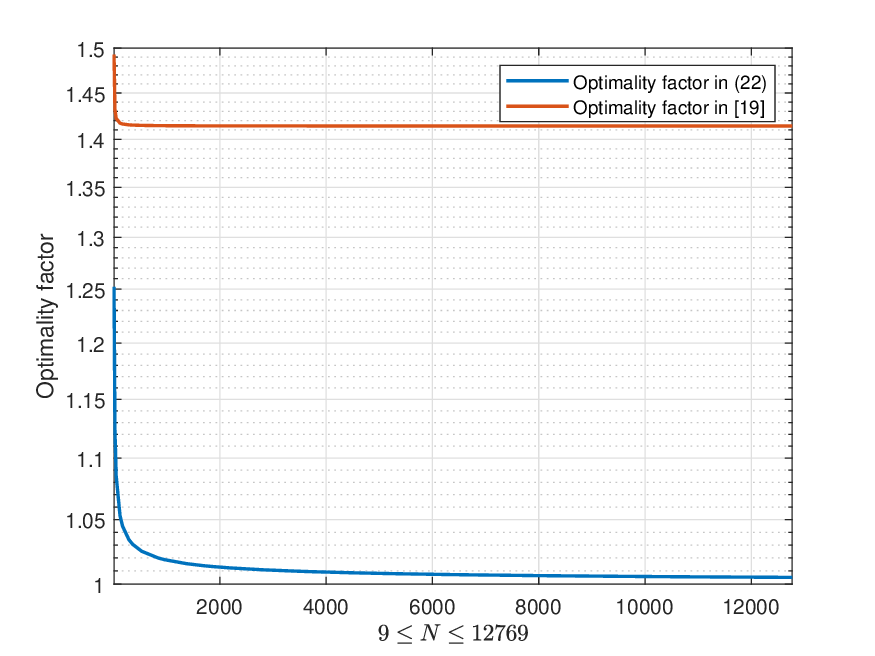}
 	\caption{Comparison between the optimality factors (\ref{rho1}) and
 		\cite{shen2024} with respect to $(p^n, p^n, p^n-1, p^n, \Pi)$-DRCS for $3 \leq  p < 120,$ $n=2$ and $\Pi=(-p^n+1,p^n-1)\times (-p^n+1,p^n-1)$.}
 	\label{fig:ex2}
 \end{figure}
  	                                                                                  
\section{Concluding Remarks}
In this paper, we have introduced a novel aperiodic AF lower bound for unimodular DRCSs, which depends on several factors, including the set size \( K \), the number of channels \( M \), the sequence length \( N \), the maximum delay shift \( Z_x \), the maximum Doppler shift \( Z_y \), and the weight vector \( \mathbf{w} \). We have analyzed three types of weight vectors, demonstrating that many established aperiodic bounds, such as the Welch bound, Sarwate bound, Levenshtein bound, Peng-Fan bound, Liu-Guan-Mow bound, and Meng-Guan-Ge-Liu-Fan bound, can be regarded as special cases of our proposed aperiodic AF lower bound.
Additionally, we have constructed a class of DRCSs with small alphabets based on quasi-Florentine rectangles and Butson-type Hadamard matrices, which are asymptotically optimal with respect to the derived aperiodic DRCS set bounds. 
The optimality factor of the proposed DRCS sets in relation to the Shen-Yang-Zhou-Liu-Fan bound converges to $\sqrt{2}$. In contrast, the optimality factor with respect to the proposed bound converges to 1. This indicates that the proposed bound is tighter than the Shen-Yang-Zhou-Liu-Fan bound.

\bibliographystyle{ieeetr}
\bibliography{ref_lb.bib}

\begin{thebibliography}{10}

\bibitem{golay1961complementary}
M.~Golay, ``Complementary series,'' {\em IRE IEEE Trans. Inf. Theory}, vol.~7,
  no.~2, pp.~82--87, 1961.

\bibitem{tseng1972complementary}
C.-C. Tseng and C.~Liu, ``Complementary sets of sequences,'' {\em IEEE Trans.
  Inf. Theory}, vol.~18, no.~5, pp.~644--652, 1972.

\bibitem{channeles2001complementary}
P.~Spasojevic and C.~N. Georghiades, ``Complementary sequences for {ISI}
  channel estimation,'' {\em IEEE Trans. Inf. Theory}, vol.~47, no.~3,
  pp.~1145--1152, 2001.

\bibitem{mccdma2001multicarrier}
H.-H. Chen, J.-F. Yeh, and N.~Suehiro, ``A multicarrier {CDMA} architecture
  based on orthogonal complementary codes for new generations of wideband
  wireless communications,'' {\em IEEE Commun. Mag.}, vol.~39, no.~10,
  pp.~126--135, 2001.

\bibitem{PAPR2021encoding}
A.~{\c{S}}ahin, ``Encoding and decoding with partitioned complementary
  sequences for low-{PAPR} {OFDM},'' {\em IEEE Trans. Wireless Commun.},
  vol.~21, no.~4, pp.~2561--2572, 2021.

\bibitem{Doppler2008doppler}
A.~Pezeshki, A.~R. Calderbank, W.~Moran, and S.~D. Howard, ``Doppler resilient
  {G}olay complementary waveforms,'' {\em IEEE Trans. Inf. Theory}, vol.~54,
  no.~9, pp.~4254--4266, 2008.

\bibitem{Doppler2024complementary}
Q.~Wang, H.~Yang, L.~Wu, L.~Zhang, Y.~Xia, X.~Fu, and C.~Tan, ``Complementary
  coding-based waveform design for broadband acoustic {D}oppler current
  profilers,'' {\em IEEE Trans. Veh. Technol.}, vol.~73, no.~7, pp.~9398--9410,
  2024.

\bibitem{ISAC2017ieee}
P.~Kumari, J.~Choi, N.~Gonz{\'a}lez-Prelcic, and R.~W. Heath, ``{IEEE} 802.11
  ad-based radar: {A}n approach to joint vehicular communication-radar
  system,'' {\em IEEE Trans. Veh. Technol.}, vol.~67, no.~4, pp.~3012--3027,
  2017.

\bibitem{ISACl2020doppler}
G.~Duggal, S.~Vishwakarma, K.~V. Mishra, and S.~S. Ram, ``Doppler-resilient
  802.11 ad-based ultrashort range automotive joint radar-communications
  system,'' {\em IEEE Trans. Aerosp. Electron. Syst.}, vol.~56, no.~5,
  pp.~4035--4048, 2020.

\bibitem{liu2013tighter}
Z.~Liu, Y.~L. Guan, and W.~H. Mow, ``A tighter correlation lower bound for
  quasi-complementary sequence sets,'' {\em IEEE Trans. Inf. Theory}, vol.~60,
  no.~1, pp.~388--396, 2013.

\bibitem{welch1974lower}
L.~Welch, ``Lower bounds on the maximum cross correlation of signals
  (corresp.),'' {\em IEEE Trans. Inf. Theory}, vol.~20, no.~3, pp.~397--399,
  1974.

\bibitem{levenshtein1999new}
V.~I. Levenshtein, ``New lower bounds on aperiodic crosscorrelation of binary
  codes,'' {\em IEEE Trans. Inf. Theory}, vol.~45, no.~1, pp.~284--288, 1999.

\bibitem{ye22}
Z.~Ye, Z.~Zhou, P.~Fan, Z.~Liu, X.~Lei, and X.~Tang, ``Low ambiguity zone:
  {T}heoretical bounds and {D}oppler-resilient sequence design in integrated
  sensing and communication systems,'' {\em IEEE J. Sel. Areas Commun.},
  vol.~40, no.~6, pp.~1809--1822, 2022.

\bibitem{ding13}
C.~Ding, K.~Feng, R.~Feng, M.~Xiong, and A.~Zhang, ``Unit time-phase signal
  sets: {B}ounds and constructions,'' {\em Cryptogr. Commun.}, vol.~5,
  pp.~209--227, 2013.

\bibitem{duggal2020doppler}
G.~Duggal, S.~Vishwakarma, K.~V. Mishra, and S.~S. Ram, ``Doppler-resilient
  802.11 ad-based ultrashort range automotive joint radar-communications
  system,'' {\em IEEE Trans. Aerosp. Electron. Syst.}, vol.~56, no.~5,
  pp.~4035--4048, 2020.

\bibitem{kumari}
P.~Kumari, J.~Choi, N.~Gonz{\'a}lez-Prelcic, and R.~W. Heath, ``{IEEE} 802.11
  ad-based radar: {A}n approach to joint vehicular communication-radar
  system,'' {\em IEEE Trans. Veh. Technol.}, vol.~67, no.~4, pp.~3012--3027,
  2017.

\bibitem{meng2024new}
L.~Meng, Y.~L. Guan, Y.~Ge, Z.~Liu, and P.~Fan, ``New lower bounds on aperiodic
  ambiguity function of unimodular sequences,'' {\em arXiv preprint
  arXiv:2402.00455}, 2024.

\bibitem{tang2014construction}
J.~Tang, N.~Zhang, Z.~Ma, and B.~Tang, ``Construction of doppler resilient
  complete complementary code in mimo radar,'' {\em IEEE Trans. Signal
  Process.}, vol.~62, no.~18, pp.~4704--4712, 2014.

\bibitem{zhu2017range}
J.~Zhu, X.~Wang, X.~Huang, S.~Suvorova, and B.~Moran, ``Range sidelobe
  suppression for using golay complementary waveforms in multiple moving target
  detection,'' {\em Signal Process.}, vol.~141, pp.~28--31, 2017.

\bibitem{zhu2019alternative}
J.~Zhu, N.~Chu, Y.~Song, S.~Yi, X.~Wang, X.~Huang, and B.~Moran, ``Alternative
  signal processing of complementary waveform returns for range sidelobe
  suppression,'' {\em Signal Process.}, vol.~159, pp.~187--192, 2019.

\bibitem{wang2021quasi}
J.~Wang, P.~Fan, Z.~Zhou, and Y.~Yang, ``Quasi-orthogonal z-complementary pairs
  and their applications in fully polarimetric radar systems,'' {\em IEEE
  Trans. Inf. Theory}, vol.~67, no.~7, pp.~4876--4890, 2021.

\bibitem{shen2024}
B.~Shen, Y.~Yang, Z.~Zhou, Z.~Liu, and P.~Fan, ``Doppler resilient
  complementary sequences: {T}heoretical bounds and optimal constructions,''
  {\em submitted to IEEE Trans. Inf. Theory}, 2024.

\bibitem{Avik2024}
A.~R. Adhikary, H.~Zhang, Z.~Zhou, Q.~Wang, and S.~Mesnager, ``Quasi
  complementary sequence sets: {N}ew bounds and optimal constructions via
  quasi-florentine rectangles,'' {\em IEEE Trans. Inf. Theory}, 2025.

\bibitem{had1973complex}
J.~Wallis, ``Complex hadamard matrices,'' {\em Linear and Multilinear Algebra},
  vol.~1, no.~3, pp.~257--272, 1973.

\bibitem{hadc2006complex}
W.~Bruzda, W.~Tadej, and K.~{\.Z}yczkowski, ``Complex hadamard matrices-a
  catalog (since 2006),'' {\em
  [online]https://chaos.if.uj.edu.pl/~karol/hadamard/index.html}, 2006.

\bibitem{liu2021tighter}
B.~Liu, Z.~Zhou, and U.~Parampalli, ``A tighter correlation lower bound for
  quasi-complementary sequence sets with low correlation zone,'' {\em IEICE
  Trans. Fundam. Electron. Commun. Comput. Sci.}, vol.~104, no.~2,
  pp.~392--398, 2021.

\bibitem{peng2004generalised}
D.~Peng and P.~Fan, ``Generalised sarwate bounds on the aperiodic correlation
  of sequences over complex roots of unity,'' {\em in IEE Proc. Commun.},
  vol.~151, no.~4, pp.~375--382, 2004.

\bibitem{berlekamp2015algebraic}
E.~R. Berlekamp, {\em Algebraic coding theory (revised edition)}.
\newblock World Scientific, 2015.

\bibitem{liu2013new}
Z.~Liu, U.~Parampalli, Y.~L. Guan, and S.~Bozta{\c{s}}, ``A new weight vector
  for a tighter levenshtein bound on aperiodic correlation,'' {\em IEEE Trans.
  Inf. Theory}, vol.~60, no.~2, pp.~1356--1366, 2013.

\bibitem{avik2021asymptotically}
A.~R. Adhikary, Y.~Feng, Z.~Zhou, and P.~Fan, ``Asymptotically optimal and
  near-optimal aperiodic quasi-complementary sequence sets based on florentine
  rectangles,'' {\em IEEE Trans. Commun.}, vol.~70, no.~3, pp.~1475--1485,
  2021.

\end{thebibliography}

\end{document}